\documentclass[paper=a4]{scrartcl}

\usepackage[utf8]{inputenc}
\usepackage[T1]{fontenc}
\usepackage[english]{babel}
\usepackage[noadjust]{cite}
\usepackage{amsmath}
\usepackage{microtype}
\usepackage{mathptmx}
\usepackage[scaled=.90]{helvet}
\usepackage{courier}

\usepackage{enumerate}
\usepackage{xspace}
\usepackage[hyphens]{url} 
\usepackage{hyperref}
\usepackage{xcolor}
\usepackage{tikz}
\usetikzlibrary{arrows}
\usetikzlibrary{shapes}
\usetikzlibrary{positioning}
\usepackage{algpseudocode}
\usepackage{algorithm}
\usepackage{paralist}

\usepackage{etoolbox}

\usepackage{amssymb}

\newtoggle{long}
\toggletrue{long} 

\newtheorem{theorem}{Theorem}
\newtheorem{lemma}{Lemma}

\newtheorem{corollary}{Corollary}
\newcommand{\qed}{}
\newcommand{\mqed}{\hfill$\Box$}
\newlength{\proofpostskipamount}\newlength{\proofpreskipamount}
\newenvironment{proof}%
               {\par\vspace{\proofpreskipamount}\noindent{\textbf{Proof:}}\hspace{0.5em}}
               {\nopagebreak%
                \strut\nopagebreak%
                \hspace{\fill}\mqed\par\vspace{\proofpostskipamount}\noindent}

\newcommand{\abs}[1]{| #1 |}

\newcommand{\C}{\ensuremath{\mathcal{C}}}
\newcommand{\pC}{\widehat{C}}
\newcommand{\sset}[1]{\{ #1 \}}

\newcommand{\edge}[2]{\ensuremath{#1\,#2}}

\graphicspath{{./pictures/}}

\title{Certifying 3-Edge-Connectivity}
\author{Kurt Mehlhorn \and Adrian Neumann \and Jens M. Schmidt}

\begin{document}
\maketitle
\begin{abstract} We present a certifying algorithm that tests
  graphs for 3-edge-connectivity; the algorithm works in linear time.  If the input graph is not
  3-edge-connected, the algorithm returns a 2-edge-cut. If it is
  3-edge-connected, it returns a construction sequence that
  constructs the input graph from the graph with two vertices and three parallel edges using
  only operations that (obviously) preserve 3-edge-connectivity.

Additionally, we show how to compute and certify the $3$-edge-connected components and a cactus representation of the $2$-cuts in linear time. For $3$-vertex-connectivity, we show how to compute the $3$-vertex-connected components of a $2$-connected graph.
\end{abstract}

\section{Introduction}
Advanced graph algorithms answer complex yes-no questions such as ``Is this graph planar?'' or ``Is this graph $k$-vertex-connected?''. These algorithms are not only nontrivial to implement, it is also difficult to test their implementations extensively, as usually only small test sets are available. It is hence possible that bugs persist unrecognized for a long time. An example is the implementation of the linear time planarity test of Hopcroft and Tarjan~\cite{HT} in LEDA~\cite{Mehlhorn1999}. A bug in the implementation was discovered only after two years of intensive use.

\emph{Certifying algorithms}~\cite{McConnell2011} approach this problem by computing an additional \emph{certificate} that proves the correctness of the answer. This may, e.g., be either a 2-coloring or an odd cycle for testing bipartiteness, or either a planar embedding or a Kuratowski subgraph for testing planarity. Certifying algorithms are designed such that checking the correctness of the certificate is substantially simpler than solving the original problem. Ideally, checking the correctness is so simple that the implementation of the checking routine allows for a formal verification~\cite{FormalVerification,Verification-CertComps-AutoCorres-Simpl}.

Our main result is a linear time certifying algorithm for $3$-edge-connectivity based on a result of Mader~\cite{Mader1978}. He showed that every 3-edge-connected graph can be obtained from $K_2^3$, the graph consisting of two vertices and three parallel edges, by a sequence of three simple operations that each introduce one edge and, trivially, preserve 3-edge-connectivity. We show how to compute such a sequence in linear time for $3$-edge-connected graphs. If the input graph is not $3$-edge-connected, a $2$-edge-cut is computed. The previous algorithms~\cite{Galil1991,Nagamochi1992a,Taoka1992,Tsin2007,Tsin2009} for deciding 3-edge-connectivity are not certifying; they deliver a 2-edge-cut for graphs that are not 3-edge-connected but no certificate in the yes-case.

Our algorithm is path-based~\cite{Gabow2000}. It uses the concept of a \emph{chain decomposition} of a graph introduced in~\cite{Schmidt2010b} and used for certifying $1$- and $2$-vertex and $2$-edge-connectivity in~\cite{Schmidt2013a} and for certifying $3$-vertex connectivity in~\cite{Schmidt2013}. 
A chain decomposition is a special ear decomposition~\cite{Lovasz1985}. 
We use chain decompositions to certify 3-edge-connectivity in linear time. Thus, chain decompositions form a common framework for certifying $k$-vertex- and $k$-edge-connectivity for $k \le 3$ in linear time. We use many techniques from~\cite{Schmidt2013}, but in a simpler form. Hence our paper may also be used as a gentle introduction to the 3-vertex-connectivity algorithm in~\cite{Schmidt2013}.

We state Mader's result in Section~\ref{preliminaries} and introduce chain decompositions in Section~\ref{sec:chain decomposition}. In Section~\ref{Chains as Mader-paths} we show that chain decompositions can be used as a basis for Mader's construction. This immediately leads to an $O((m + n) \log(m + n))$ certifying algorithm (Section~\ref{A First Algorithm}). The linear time algorithm is then presented in Sections~\ref{A Classification of Chains} and~\ref{sec:linear time alg}. In Section~\ref{sec:verify mader} we discuss the verification of Mader construction sequences.

The mincuts in a graph can be represented succinctly by a cactus representation~\cite{Dinits-Karzanov-Lomonosov,Nagamochi-Ibaraki-Book,Fleiner2009}; see Section~\ref{Sec: Cactus Representation}. The 3-edge-connected components of a graph are the maximal subsets of the vertex set such that any two vertices in the subset are connected by three edge-disjoint paths. These paths are not necessarily contained in the subset. 

Our algorithm can be used to turn any algorithm for computing 3-edge-connected components into a certifying algorithm for computing 3-edge-connected components and the cactus representation of 2-cuts (Section~\ref{Sec: Cactus Representation}). An extension of our algorithm computes the 3-edge-connected components and the cactus representation directly (Section~\ref{Computing a Cactus}).  A similar technique can be used to extend the 3-vertex-connectivity algorithm in~\cite{Schmidt2013} to an algorithm for computing 3-vertex-connected components.

\section{Related Work}
Deciding $3$-edge-connectivity is a well researched problem, with applications in diverse fields such as bioinformatics~\cite{dehne2006cluster} and quantum chemistry~\cite{corcoran2006perfect}. Consequently, there are many different linear time solutions known~\cite{Galil1991,Nagamochi1992a,Taoka1992,Tsin2007,Tsin2009,Nagamochi-Ibaraki-Book}. None of them is certifying. All but the first algorithm also compute the 3-edge-connected components. The cactus representation of a 2-edge-connected, but not 3-edge-connected graph $G$, can be obtained from $G$ by repeatedly contracting the 3-edge-connected components to single vertices~\cite{Nagamochi-Ibaraki-Book}.

The paper~\cite{McConnell2011} is a recent survey on certifying algorithms. For a linear time certifying algorithm for 3-vertex-connectivity, see~\cite{Schmidt2013} (implemented in~\cite{Neumann2011}). For general $k$, there is a randomized certifying algorithm for $k$-vertex connectivity in~\cite{Linial1988} with expected running time $O(kn^{2.5} + nk^{3.5})$. There is a non-certifying algorithm~\cite{Karger2000} for deciding $k$-edge-connectivity in time $O(m \log^3{n})$ with high probability.

In~\cite{Galil1991}, a linear time algorithm is described that transforms a graph $G$ into a graph $G'$ such that $G$ is 3-edge-connected if and only if $G'$ is 3-vertex-connected. Combined with this transformation, the certifying 3-vertex-connectivity algorithm from~\cite{Schmidt2013} certifies 3-edge-connectivity in linear time. However, that algorithm is much more complex than the algorithm given here. Moreover, we were unable to find an elegant method for transforming the certificate obtained for the 3-vertex-connectivity of $G'$ into a certificate for 3-edge-connectivity of $G$.

\section{Preliminaries}\label{preliminaries}

We consider finite undirected graphs $G$ with $n = |V(G)|$ vertices, $m = |E(G)|$ edges, no self-loops, and minimum degree three, and use standard graph-theoretic terminology from~\cite{Bondy2008}, unless stated otherwise. We use $\edge{u}{v}$ to denote an edge with endpoints $u$ and $v$.

A set of edges that leaves a disconnected graph upon deletion is called \emph{edge cut}. For $k \geq 1$, let a graph $G$ be \emph{$k$-edge-connected} if $n \ge 2$ and there is no edge cut $X \subseteq E(G)$ with $|X| < k$. Let $v \rightarrow_G w$ denote a path $P$ between two vertices $v$ and $w$ in $G$ and let $s(P)=v$ and $t(P)=w$ be the source and target vertex of $P$, respectively (as $G$ is undirected, the direction of $P$ is given by $s(P)$ and $t(P)$). Every vertex in $P \setminus \{s(P),t(P)\}$ is called an \emph{inner vertex} of $P$ and every vertex in $P$ is said to \emph{lie on} $P$.

Let $T$ be an undirected tree rooted at vertex $r$. For two vertices $x$ and $y$ in $T$, $x$ is an \emph{ancestor} of $y$ and $y$ is a \emph{descendant} of $x$ if $x \in V(r \rightarrow_T y)$, where $V(r \rightarrow_T y)$ denotes the vertex set of the path from $r$ to $y$ in $T$. If additionally $x \neq y$, $x$ is a \emph{proper} ancestor and $y$ is a \emph{proper} descendant. We write $x \le y$ ($x<y$) if $x$ is an ancestor (proper ancestor) of $y$. The parent $p(v)$ of a vertex $v$ is its immediate proper ancestor. The parent function is undefined for $r$. Let $K^m_2$ be the graph on $2$ vertices that contains exactly $m$ parallel edges.

Let \emph{subdividing an edge} $\edge{u}{v}$ of a graph $G$ be the operation that replaces $\edge{u}{v}$ with a path $uzv$, where $z$ was not previously in $G$. All 3-edge-connected graphs can be constructed using a small set of operations starting from a $K^3_2$.

\begin{figure}[htbp]
\centering
\includegraphics[width=0.8\linewidth]{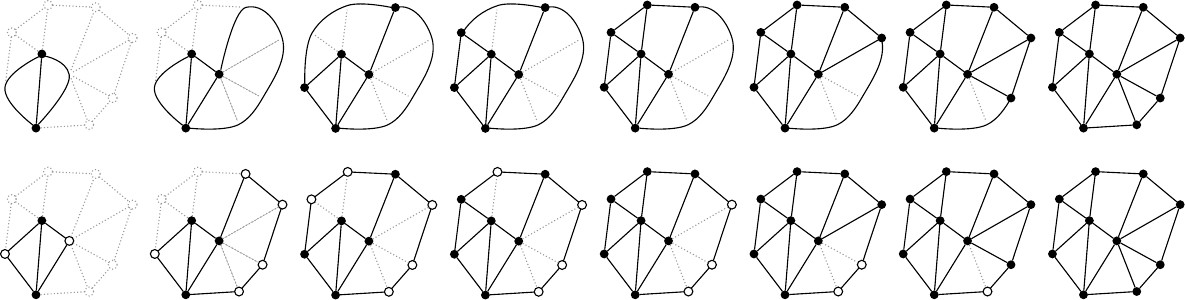}
\caption{Two ways of constructing the $3$-edge-connected graph shown in the rightmost column. The upper row shows the construction according to Theorem~\ref{Mader}. The lower row shows the construction according to Corollary~\ref{Mader-Cor}. Branch (non-branch) vertices are depicted as filled (non-filled) circles. The black edges exist already, while dotted gray vertices and edges do not exist yet. }
\label{fig:mader_construction}
\end{figure}

\begin{theorem}[Mader~\cite{Mader1978}]\label{Mader}
Every $3$-edge-connected graph (and no other graph) can be constructed from a $K_2^3$ using the following three operations:
\begin{compactitem}
\item Adding an edge (possibly parallel or a loop).
\item Subdividing an edge $\edge{x}{y}$ and connecting the new vertex to any existing vertex.
\item Subdividing two distinct edges $\edge{w}{x}$, $\edge{y}{z}$ and connecting the two new vertices.
\end{compactitem}
\end{theorem}

A subdivision $G'$ of a graph $G$ is a graph obtained by subdividing edges of $G$ zero or more times. The \emph{branch vertices} of a subdivision are the vertices with degree at least three (we call the other vertices \emph{non-branch}-vertices) and the \emph{links} of a subdivision are the maximal paths whose inner vertices have degree two. If $G$ has no vertex of degree two, the links of $G'$ are in one-to-one correspondence to the edges of $G$. Theorem~\ref{Mader} readily generalizes to subdivisions of 3-edge-connected graphs.

\begin{corollary}\label{Mader-Cor}
Every subdivision of a $3$-edge-connected graph (and no other graph) can be constructed from a
subdivision of a $K_2^3$ using the following three operations:
\begin{compactitem}
\item Adding a path connecting two branch vertices.
\item Adding a path connecting a branch vertex and a non-branch vertex.
\item Adding a path connecting two non-branch vertices lying on distinct links.
\end{compactitem}
In all three cases, the inner vertices of the path added are new vertices.
\end{corollary}

Each path that is added to a graph $H$ in the process of Corollary~\ref{Mader-Cor} is called a \emph{Mader-path} (\emph{with respect to $H$}). Note that an ear is always a Mader-path unless both endpoints lie on the same link.

Figure~\ref{fig:mader_construction} shows two constructions of a $3$-edge-connected graph, one according to Theorem~\ref{Mader} and one according to Corollary~\ref{Mader-Cor}. In this paper, we show how to find the Mader construction sequence according to Corollary~\ref{Mader-Cor} for a 3-edge-connected graph in linear time. Such a construction is readily turned into one according to Theorem~\ref{Mader}.

\section{Chain Decompositions}\label{sec:chain decomposition}

We use a very simple decomposition of graphs into cycles and paths. The decomposition was previously used for linear-time tests of $2$-vertex- and $2$-edge-connectivity~\cite{Schmidt2013a} and $3$-vertex-connectivity~\cite{Schmidt2013}. In this paper we show that it can also be used to find a Mader's construction for a $3$-edge-connected graph. We define the decomposition algorithmically; a similar procedure that serves for the computation of low-points can be found in~\cite{Ramachandran1993}.

Let $G$ be a connected graph without self-loops and let $T$ be a depth-first search tree of $G$. Let $r$ be the root of $T$. We orient tree-edges towards the root and back-edges away from the root, i.e., $v < u$ for an oriented tree-edge $\edge uv$ and $x < y$ for an oriented back-edge $\edge xy$. 

We decompose 
$G$ into a set $\C = \{C_1,\ldots,C_{|\C|}\}$ of cycles and paths, called \emph{chains}, by applying the following procedure for each vertex $v$ in the order in which they were discovered during the DFS:
First, we declare $v$ visited (initially, no vertex is visited), if not already visited before. Then, for every back-edge $\edge{v}{w}$, we traverse the path $w \rightarrow_{T} r$ until a vertex $x$ is encountered that was visited before; $x$ is a descendant of $v$. The traversed subgraph $\edge{v}{w} \cup (w \rightarrow_{T} x)$ forms a new \emph{chain} $C$ with $s(C)=v$ and $t(C)=x$. All inner vertices of $C$ are declared visited. Observe that $s(C)$ and $t(C)$ are already visited when the construction of the chain starts.

Figure~\ref{fig:chains} illustrates these definitions. Since every back-edge defines one chain, there are precisely $m-n+1$ chains. We number the chains in the order of their construction.
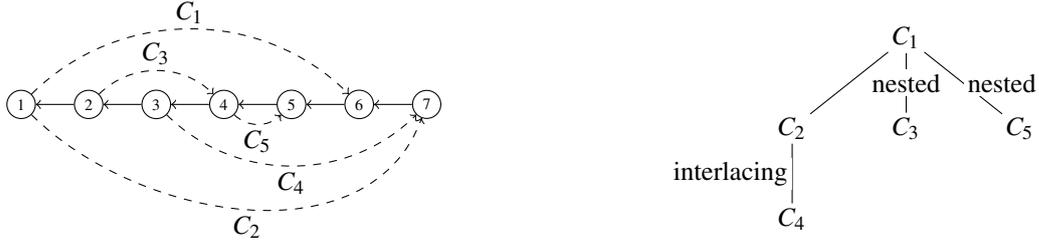
\begin{figure}
\centering
\begin{minipage}{0.4\linewidth}
\begin{tikzpicture}[node distance = 0.5cm, inner sep = 2pt]
\tikzstyle{every node}=[circle, draw]
\tiny
\node (1) {1};
\foreach \i in {2,3,4,5,6,7} {
	\pgfmathtruncatemacro{\x}{\i-1};
	\node (\i) [right= of \x] {\i};
	\draw[->] (\i) -- (\x);
}
\small
\tikzstyle{every node}=[]
\draw[->, dashed] (1)
	edge [out = 45, in=-225] node [above] {$C_1$} (6)
	edge [out=-45, in = 250] node [below] {$C_2$} (7);
\draw[->, dashed] (2) edge [out = 45, in=-225] node [above] {$C_3$} (4);
\draw[->, dashed] (3) edge [out = -45, in = 225] node [below] {$C_4$} (7);
\draw[->, dashed] (4) edge [out = -45, in = 225] node [below] {$C_5$} (5);
\end{tikzpicture}
\end{minipage}
\hfill
\begin{minipage}{0.4\linewidth}
\begin{tikzpicture}[inner sep= 1.5pt]
\small
 \node {$C_1$}
 [level distance = 12mm]
    child {
    	node {$C_2$}
    	child {
        	node {$C_4$}
        	edge from parent node[left, fill = white] {interlacing}
        }
        edge from parent node[left] {}
    }
    child {
    	node {$C_3$}
    	edge from parent node[midway, fill = white] {nested}
    }
    child {
    	node {$C_5$}
    	edge from parent node[right, fill = white] {nested}
    };
\end{tikzpicture}
\end{minipage}
\caption{
The left side of the figure shows a DFS tree with a chain decomposition; tree-edges are solid and back-edges are dashed. $C_1$ is (\edge{1}{6},\edge{6}{5},\edge{5}{4},\edge{4}{3},\edge{3}{2},\edge{2}{1}), $C_2$ is (\edge{1}{7},\edge{7}{6}), $C_3$ is (\edge{2}{4}), $C_4$ is (\edge{3}{7}), and $C_5$ is (\edge{4}{5}). $C_3$ and $C_5$ are nested children of $C_1$ and $C_4$ is an interlacing child of $C_2$. Also, $s(C_4)$ s-belongs to $C_1$.} 
\label{fig:chains}
\end{figure}

We call $\C$ a \emph{chain decomposition}. It can be computed in time $O(n+m)$. For 2-edge-connected graphs the term decomposition is justified by Lemma~\ref{lem:2connectivity}.

\begin{lemma}[\cite{Schmidt2013a}]\label{lem:2connectivity}
Let $\C$ be a chain decomposition of a graph $G$. Then $G$ is 2-edge-connected if and only if $G$ is connected and the chains in $\C$ partition $E(G)$.
\end{lemma}

Since the condition of Lemma~\ref{lem:2connectivity} is easily checked during the chain decomposition, we assume from now on that $G$ is 2-edge-connected. Then $\C$ partitions $E(G)$ and the first chain $C_1$ is a cycle containing $r$ (since there is a back-edge incident to $r$). We say that $r$ \emph{strongly belongs (s-belongs)} to the first chain and any vertex $v \neq r$ \emph{s-belongs} to the chain containing the edge $\edge{v}{p(v)}$. We use s-belongs instead of belongs since a vertex can belong to many chains when chains are viewed as sets of vertices.

We can now define a parent-tree on chains. The first chain $C_1$ is the root. For any chain $C \neq C_1$, let the \emph{parent} $p(C)$ of $C$ be the chain to which $t(C)$ s-belongs. We write $C \le D$ ($C<D$) for chains $C$ and $D$ if $C$ is an ancestor (proper ancestor) of $D$ in the parent-tree on chains.

The following lemma summarizes important properties of chain decompositions.

\begin{lemma}\label{facts about chain decomposition}
Let $\{C_1,\ldots,C_{m-n+1}\}$ be a chain decomposition of a 2-edge-connected graph $G$ and let $r$ be the root of the DFS-tree. Then
\begin{compactenum}[(1)]
\item \label{s not above t} For every chain $C_i$, $s(C_i) \le t(C_i)$.
\item \label{properties of parent chain} Every chain $C_i$, $i \ge 2$, has a parent chain $p(C_i)$. We have $s(p(C_i)) \le s(C_i)$ and $p(C_i) = C_j$ for some $j < i$.
\item \label{t-parent precedes t} For $i \ge 2$: If $t(C_i) \not= r$, $t(p(C_i)) < t(C_i)$. If $t(C_i) = r$, $t(p(C_i)) = t(C_i)$.
\item \label{parent on chains follows parent on nodes} If $u \le v$, $u$ s-belongs to $C$, and $v$ s-belongs to $D$ then $C \le D$.
\item \label{t(D) above u} If $u \le t(D)$ and $u$ s-belongs to $C$, then $C \le D$.
\item \label{s(C) s-belongs to earlier chain} For $i \ge 2$: $s(C_i)$ s-belongs to a chain $C_j$ with $j < i$.
\end{compactenum}
\end{lemma}
\begin{proof}
(\ref{s not above t}) to (\ref{t-parent precedes t}) follow from the discussion preceding the Lemma and the construction of the chains. We turn to (\ref{parent on chains follows parent on nodes}). Consider two vertices $u$ and $v$ with $u \le v$ and let $u$ s-belong to $C$ and let $v$ s-belong to $D$. Then $C \le D$, as the following simple induction on the length of the tree path from $u$ to $v$ shows. If $u = v$, $C=D$ by the definition of s-belongs. So assume $u$ is a proper ancestor of $v$. Since $v$ s-belongs to $D$, by definition $v \not= t(D)$ and $\edge{v}{p(v)}$ is contained in $D$. Let $D'$ be the chain to which $p(v)$ s-belongs. By induction hypothesis, $C \le D'$. Also, either $D = D'$ (if $p(v)$ s-belongs to $D$) or $D' = p(D)$ (if $p(v) = t(D)$) and hence $p(v)$ s-belongs to $p(D)$. In either case $C \le D$.

Claim (\ref{t(D) above u}) is an easy consequence of (\ref{parent on chains follows parent on nodes}). If $t(D) = r$, $C = C_1$, and the claim follows. If $t(D) \not= r$, $t(D)$ s-belongs to $p(D)$. Thus, $C \le p(D)$ by  (\ref{parent on chains follows parent on nodes}).

The final claim is certainly true for each $C_i$ with $s(C_i) = r$. So assume $s(C_i) > r$ and let $y = p(s(C_i))$. Since $G$ is 2-edge-connected, there is a back-edge $\edge uv$ with $u \le y$ and $s(C_i) \le v$. It induces a chain $C_k$ with $k < i$ and hence $s(C_i) y$ is contained in a chain $C_j$ with $j \le k$. \qed
\end{proof}

\section{Chains as Mader-paths}\label{Chains as Mader-paths}

We show that, assuming that the input graph is 3-edge-connected, there are two chains that form a subdivision of a $K_2^3$, and that the other chains of the chain decomposition can be added one by one such that each chain is a Mader-path with respect to the union of the previously added chains. We will also show that chains can be added parent-first, i.e., when a chain is added, its parent was already added. In this way the current graph $G_c$ consisting of the already added chains is always \emph{parent-closed}. We will later show how to compute this ordering efficiently. We will first give an $O((n + m) \log (n + m))$ algorithm and then a linear time algorithm.


Using the chain decomposition, we can identify a $K_2^3$ subdivision in the graph as follows. We may assume that the first two back-edges explored from $r$ in the DFS have their other endpoint in the same subtree $T'$ rooted at some child of $r$. The first chain $C_1$ forms a cycle. The vertices in $C_1 \setminus r$ are then contained in $T'$. By assumption, the second chain is constructed by another back-edge that connects $r$ with a vertex in $T'$. If there is no such back-edge, the tree edge connecting $r$ and the root of $T'$ and the back edge from $r$ into $T'$ form a $2$-edge cut. 
Let $x = t(C_2)$. Then $C_1\cup C_2$ forms a $K_2^3$ subdivision with branch vertices $r$ and $x$. The next lemma derives properties of parent-closed unions of chains.

\begin{lemma}\label{parent-closed union of chains}
Let $G_c$ be a parent-closed union of chains that contains $C_1$ and
$C_2$. Then
\begin{compactenum}[(1)]
\item For any vertex $v \not= r$ of $G_c$, the edge $\edge{v}{p(v)}$ is contained in $G_c$, i.e., the set of vertices of $G_c$ is a parent-closed subset of the DFS-tree.
\item $s(C)$ and $t(C)$ are branch vertices of $G_c$ for every chain $C$ contained in $G_c$.
\item Let $C$ be a chain that is not in $G_c$ but a child of some chain in $G_c$. Then $C$ is an ear with respect to $G_c$ and the path $t(C)\rightarrow_T s(C)$ is contained in $G_c$. $C$ is a Mader-path (i.e., the endpoints of $C$ are not inner vertices of the same link of $G_c$) with respect to $G_c$ if and only if there is a branch vertex on $t(C)\rightarrow_T s(C)$.
\end{compactenum}
\end{lemma}
\begin{proof}
(1): Let $v \not= r$ be any vertex of $G_c$. Let $C$ be a chain in $G_c$ containing the vertex $v$. If $C$ also contains $\edge{v}{p(v)}$ we are done. Otherwise, $v = t(C)$ or $v = s(C)$. In the first case, $v$ s-belongs to $p(C)$, in the second case $v$ s-belongs to some $C'\leq C$ by Lemma~\ref{facts about chain decomposition}.(\ref{parent on chains follows parent on nodes}). Hence, by parent-closedness, $\edge{v}{p(v)}$ is an edge of $G_c$.

(2): Let $C$ be any chain in $G_c$. Since $C_1$ and $C_2$ form a $K^3_2$, $r$ and $x=t(C_2)$ are branch vertices. If $s(C) \not= r$, the edge $\edge{s(C)}{p(s(C))}$ is in $G_c$ by (1), the back-edge $\edge{s(C)}{v}$ inducing $C$ is in $G_c$, and the path $v \rightarrow_T s(C)$ is in $G_c$ by (1). Thus $s(C)$ has degree at least three. If $t(C) \not\in \sset{r,x}$, let $\pC$ be the chain to which $t(C)$ s-belongs, i.e. $\pC$ is the parent of $C$. As $G_c$ is parent-closed $\pC$ is contained in $G_c$. By the definition of s-belongs, $t(C)$ has degree two on the chain $\pC$. Further, it has degree one on the chain $C$. Since chains are edge-disjoint, it has degree at least three in $G_c$.

(3) We first observe that $t(C)$ and $s(C)$ belong to $G_c$. For $t(C)$, this holds since $t(C)$ s-belongs to $p(C)$ and $p(C)$ is part of $G_c$ by assumption. For $s(C)$, this follows from $s(C) \le t(C)$ and (1). No inner vertex $u$ of $C$ belongs to $G_c$, because otherwise the edge $\edge{u}{p(u)}$ would belong to $G_c$ by (1), which implies that $C$ would belong to $G_c$, as $G_c$ is a union of chains. Thus $C$ is an ear with respect to $G_c$, i.e., it is disjoint from $G_c$ except for its endpoints. Moreover, the path $t(C) \rightarrow_T s(C)$ belongs to $G_c$ by (1).

If there is no branch vertex on $t(C)\rightarrow_T s(C)$, the vertices $t(C)$ and $s(C)$ are inner vertices of the same link of $G_c$ and hence $C$ is not a Mader-path with respect to $G_c$. If there is a branch vertex on $t(C)\rightarrow_T s(C)$, the vertices $t(C)$ and $s(C)$ are inner vertices of two distinct links of $G_c$ and hence $C$ is a Mader-path with respect to $G_c$.
\qed
\end{proof}

We can now prove that chains can always be added in parent-first order. For a link $L$, each edge in $L$ that is incident to an end vertex of $L$ is called an \emph{extremal} edge of $L$.

\begin{theorem}\label{order exists} Let $G$ be a graph and let $G_c$ be a parent-closed union of chains such that no child of a chain $C \in G_c$ is a Mader-path with respect to $G_c$ and there is at least one such chain. Then the extremal edges of every link of length at least two in $G_c$ are a 2-cut in $G$.
\end{theorem}
\begin{proof} Assume otherwise. Then there is a parent-closed union $G_c$ of
  chains such that no child of a chain in $G_c$ is a Mader-path with respect to $G_c$ and there is at least one such chain outside of $G_c$, but for every link in $G_c$ the extremal edges are not a cut in $G$.

Consider any link $L$ of $G_c$. %
%
%
Since the extremal edges of $L$ do not form a 2-cut, there is a path in $G-G_c$ connecting an inner vertex on $L$ with a vertex that is either a branch vertex of $G_c$ or a vertex on a link of $G_c$ different from $L$. Let $P$ be such a path of minimum length. By minimality, no inner vertex of $P$ belongs to $G_c$. Note that $P$ is a Mader-path with respect to $G_c$. We will show that at least one edge of $P$ belongs to a chain $C$ with $p(C)\in G_c$ and that $C$ can be added, contradicting our choice of $G_c$.

Let $a$ and $b$ be the endpoints of $P$, and let $z$ be the lowest common ancestor of all points in $P$. Since a DFS generates only tree- and back-edges, $z$ lies on $P$. Since $z \le a$ and the vertex set of $G_c$ is a parent-closed subset of the DFS-tree, $z$ belongs to $G_c$. Thus $z$ cannot be an inner vertex of $P$ and hence is 
equal to $a$ or $b$. Assume w.l.o.g.~that $z = a$. All vertices of $P$ are descendants of $a$. We view $P$ as oriented from $a$ to $b$.

Since $b$ is a vertex of $G_c$, the path $b\rightarrow_T a$ is part of $G_c$ by Lemma~\ref{facts about chain decomposition} and hence no inner vertex of $P$ lies on this path. Let $\edge{a}{v}$ be the first edge on $P$. The vertex $v$ must be a descendant of $b$ as otherwise the path $v\rightarrow _P b$ would contain a cross-edge, i.e.\ an edge between different subtrees. Hence $\edge{a}{v}$ is a back-edge. Let $D$ be the chain that starts with the edge $\edge{a}{v}$. $D$ does not belong to $G_c$, as no edge of $P$ belongs to $G_c$.

We claim that $t(D)$ is a proper descendant of $b$ or $D$ is a Mader-path with respect to $G_c$. Since $v$ is a descendant of $b$ and $t(D)$ is an ancestor of $v$, $t(D)$ is either a proper descendant of $b$, equal to $b$, or a proper ancestor of $b$. We consider each case separately.

If $t(D)$ were a proper ancestor of $b$ the edge $\edge{b}{p(b)}$ would belong to $D$ and hence $D$ would be part of $G_c$, contradicting our choice of $P$. If $t(D)$ is equal to $b$ then $D$ is a Mader-path with respect to $G_c$. This leaves the case that $t(D)$ is a proper descendant of $b$.

Let $\edge yx$ be the last edge on the path $t(D)\rightarrow_T b$ that is not in $G_c$ and let $D^*$ be the chain containing $\edge yx$. Then $D^* \le D$ by Lemma~\ref{facts about chain decomposition}.(\ref{t(D) above u}) (applied with $C = D^*$ and $u = y$) and hence $s(D^*) \le s(D) \le a$ by part (\ref{parent on chains follows parent on nodes}) of the same lemma. Also $t(D^*) = x$. Since $x=t(D^*)\in G_c$, $p(D^*)\in G_c$.

As $a$ and $b$ are not inner vertices of the same link, the path $x=t(D^*) \rightarrow_T b \rightarrow_T a \rightarrow_T s(D^*)$ contains a branch vertex. Thus $D^*$ is a Mader-path by Lemma~\ref{parent-closed union of chains}.
\qed
\end{proof}
\begin{corollary}\label{greedy works}  If $G$ is 3-edge-connected, chains can be greedily added in parent-first order.
\end{corollary}
\begin{proof} If we reach a point where not all chains are added, but we can not proceed in a greedy fashion, by Theorem~\ref{order exists} we find a cut in $G$.
\end{proof}

\section{A First Algorithm}\label{A First Algorithm}

Corollary~\ref{greedy works} gives rise to an $O((n+m) \log (n+m))$ algorithm, the Greedy-Chain-Addition Algorithm. In addition to $G$, we maintain the following data structures:
\newcommand{\cL}{\ensuremath{\mathcal{L}}\xspace}
\begin{compactitem}
\item The current graph $G_c$. Each link is maintained as a doubly linked list of vertices. Observe that all inner vertices of a link lie on the same tree path and hence are numbered in decreasing order. The vertices in $G$ are labeled \emph{inactive}, \emph{branch}, or \emph{non-branch}. The vertices in $G \setminus G_c$ are called \emph{inactive}. Every non-branch vertex stores a pointer to the link on which it lies and a list of all chains incident to it and having the other endpoint as an inner vertex of the same link. \item A list \cL of addable chains. A chain is addable if it is a Mader-path with respect to the current graph.
\item For each chain its list of children.
\end{compactitem}

We initialize $G_c$ to $C_1 \cup C_2$. It has three links, $t(C_2)\rightarrow_Tr$, $r\rightarrow _{C_1}t(C_2)$, and $r\rightarrow_{C_2} t(C_2)$. We then iterate over the children of $C_1$ and $C_2$. For each child, we check in constant time whether its endpoints are inner vertices of the same link. If so, we associate the chain with the link by inserting it into the lists of both endpoints. If not, we add the chain to the list of addable chains. The initialization process takes time $O(n + m)$.

As long as the list of addable chains is non-empty, we add a chain, say $C$. Let $u$ and $v$ be the endpoints of $C$. We perform the following actions:
\begin{compactitem}
\item If $u$ is a non-branch vertex, we make it a branch vertex. This splits the link containing it and entails some processing of the chains having both endpoints on this link.
\item If $v$ is a non-branch vertex, we make it a branch vertex. This splits the link containing it, and entails some processing of the chains having both endpoints on this link.
\item We add $C$ as a new link to $G_c$.
\item We process the children of $C$.
\end{compactitem}
We next give the details for each action.

If $u$ is a non-branch vertex, it becomes a branch vertex. Let $L$ be the link of $G_c$ containing $u$; $L$ is split into links $L_1$ and $L_2$ and the set $S$ of chains having both endpoints on $L$ is split into sets $S_1$, $S_2$ and $S_{\text{add}}$, where $S_i$ is the set of chains having both endpoints on $L_i$, $i = 1,2$, and $S_{\text{add}}$ is the set of chains that become addable (because they are incident to $u$ or have one endpoint each in $L_1$ and $L_2$). We show that we can perform the split of $L$ in time $O(1 + \abs{S_{\text{add}}} + \min(\abs{L_1} + \abs{S_1},\abs{L_2} + \abs{S_2}))$. We walk from both ends of $L$ towards $u$ in lockstep fashion. In each step we either move to the next vertex or consider one chain. Once we reach $u$ we stop. Observe that this strategy guarantees the time bound claimed above.

When we consider a chain, we check whether we can move it to the set of addable chains. If so, we do it and delete the chain from the lists of both endpoints. Once, we have reached $u$, we split the list representing the link into two. The longer part of the list retains its identity, for the shorter part we create a new list header and redirect all pointers of its elements.

Adding $C$ to $G_c$ is easy. We establish a list for the new link and let all inner vertices of $C$ point to it. The inner vertices become active non-branch vertices.

Processing the children of $C$ is also easy. For each child, we check whether both endpoints are inner vertices of $C$. If so, we insert the child into the list of its endpoints. If not, we add the child to the list of addable chains.

If \cL becomes empty, we stop. If all chains have been added, we have constructed a Mader sequence. If not all chains have been processed, there must be a link having at least one inner vertex. The first and the last edge of this link form a 2-edge-cut.

It remains to argue that the algorithm runs in time $O((n + m) \log (n + m))$. We only need to argue about the splitting process. We distribute the cost $O(1 + \abs{S_{\text{add}}} + \min(\abs{L_1} + \abs{S_1},\abs{L_2} + \abs{S_2}))$ as follows: $O(1)$ is charged to the vertex that becomes a branch vertex. All such charges add up to $O(n)$. $O(\abs{S_{\text{add}}})$ is charged to the chains that become addable. All such charges add up to $O(m)$. $O(\min(\abs{L_1} + \abs{S_1},\abs{L_2} + \abs{S_2}))$ is charged to the vertices and chains that define the minimum. We account for these charges with the following token scheme inspired by the analysis of the corresponding recurrence relation in~\cite{Me3}.

Consider a link $L$ with $k$ chains having both endpoints on $L$. We maintain the invariant that each vertex and chain owns at least $\log (\abs{L} + k)$ tokens. When a link is newly created we give $\log(n + m)$ tokens to each vertex of the link and to each chain having both endpoints on the link. In total we create $O((n+m)\log(n+m))$ tokens. Assume now that we split a link $L$ with $k$ chains into links $L_1$ and $L_2$ with $k_1$ and $k_2$ chains respectively. Then $\min(\abs{L_1} + k_1,\abs{L_2} + k_2) \le (\abs{L} + k)/2$ and hence we may take one token away from each vertex and chain of the sublink that is charged without violating the token invariant.

\begin{theorem} The Greedy-Chain-Addition algorithm runs in time $O((n + m) \log (n + m))$.
\end{theorem}

\section{A Classification of Chains}\label{A Classification of Chains}
When we add a chain in the Greedy-Chain-Addition algorithm, we also process its children. Children that do not have both endpoints as inner vertices of the chain can be added to the list of addable chains immediately. However, children that have both endpoints as inner vertices of the chain cannot be added immediately and need to be observed further until they become addable. We now make this distinction explicit by classifying chains into two types, interlacing and nested.

We classify the chains $\{C_3,\ldots C_{m-n+1}\}$ into two types. Let $C$ be a chain with parent $\pC = p(C)$. We distinguish two cases\footnote{In~\cite{Schmidt2013}, three types of chains are distinguished. What we call nested is called Type 1 there and what we call interlacing is split into Types 2 and 3 there. We do not need this finer distinction.} for $C$.
\begin{compactitem}
	\item If $s(C)$ is an ancestor of $t(\pC)$ and a descendant of $s(\pC)$, $C$ is \emph{interlacing}. We have $s(\pC)\le s(C)\le t(\pC)\le t(C)$.
	\item If $s(C)$ is a proper descendant of $t(\pC)$, $C$ is \emph{nested}. We have $s(\pC) \le t(\pC) < s(C) \le t(C)$ and $t(C)\rightarrow_T s(C)$ is contained in $\pC$.
\end{compactitem}
These cases are exhaustive as the following argument shows. Let $\edge{s(\pC)}{v}$ be the first edge on $\pC$. By Lemma~\ref{facts about chain decomposition}, $s(\pC) \le s(C) \le v$. We split the path $v\rightarrow_T s(\pC)$ into $t(\pC)\rightarrow _T s(\pC)$ and $(v\rightarrow_T t(\pC))\backslash t(\pC)$. Depending on which of these paths $s(C)$ lies on, $C$ is interlacing or nested.

The following simple observations are useful. For any chain $C \neq C_1$, $t(C)$ s-belongs to $\pC$. If $C$ is nested, $s(C)$ and $t(C)$ s-belong to $\pC$. If $C$ is interlacing, $s(C)$ s-belongs to a chain which is a proper ancestor of $\pC$ or $\pC = C_1$. The next lemma confirms that interlacing chains can be added once their parent belongs to $G_c$.

\begin{lemma}\label{interlacing are easy}
Let $G_c$ be a parent-closed union of chains that contains $C_1$ and $C_2$, let $C$ be any chain contained in $G_c$, and let $D$ be an interlacing child of $C$ not contained in $G_c$. Then $D$ is a Mader-path with respect to $G_c$.
\end{lemma}
\begin{proof} We have already shown in Lemma~\ref{parent-closed union of chains} that $D$ is an ear with respect to $G_c$, that the path $t(D) \rightarrow_T s(D)$ is part of $G_c$, and that $s(C)$ and $t(C)$ are branching vertices of $G_c$. Since $D$ is interlacing, we have $s(C) \le s(D) \le t(C) \le t(D)$. Thus $t(D) \rightarrow_T s(D)$ contains a branching vertex and hence $D$ is a Mader-path by Lemma~\ref{parent-closed union of chains}.(3). \qed
\end{proof}

\section{A Linear Time Algorithm}\label{sec:linear time alg}

According to Lemma~\ref{interlacing are easy}, interlacing chains whose parent belongs to the current graph are always Mader-paths and can be added. Nested chains have both endpoints on their parent chain and can only be added once the tree-path connecting its endpoints contains a branching point. Consider a chain nested in chain $C_i$. Which chains can help its addition by creating branching points on $C_i$? First,  interlacing chains having their source on some $C_j$ with $j \le i$, and second, chains nested in $C_i$ and their interlacing offspring having  their source on $C_i$. Chains having their source on some $C_j$ with $j > i$ cannot help because they have no endpoint on $C_i$. This observation shows that chains can be added in phases. In the $i$-th phase, we try to add all chains having their source vertex on $C_i$.




\begin{algorithm}[t]
\caption{Certifying linear-time algorithm for 3-edge connectivity.}\label{alg: main}
\begin{algorithmic}
\Procedure{Connectivity}{G=(V,E)}
\State Let $\{C_1,C_2,\ldots,C_{m-n+1}\}$ be a chain decomposition of $G$ as
described in Sect.~\ref{sec:chain decomposition};
\State Initialize $G_c$ to $C_1 \cup C_2$;
\For {$i$ from 1 to $m-n+1$}
             \Comment \emph{Phase $i$: add all chains whose source s-belongs to $C_i$}
	\State Group the chains $C$ for which $s(C)$ s-belongs to $C_i$ into
        segments;
              \State \emph{Part I of Phase $i$}: Add all segments to $G_c$ whose minimal chain is interlacing;
	\State \emph{Part II of Phase $i$}: Either find an insertion order $S_1,\ldots,S_k$ of the segments having a nested minimal chain or exhibit a 2-edge-cut and stop;
	\For {$j$ from 1 to $k$}
		\State Add the chains contained in $S_j$ parent-first;
	\EndFor
\EndFor
\EndProcedure
\end{algorithmic}
\end{algorithm}

The overall structure of the linear-time algorithm is given in Algorithm~\ref{alg: main}. An implementation in Python is available at \url{https://github.com/adrianN/edge-connectivity}. The algorithm operates in phases and maintains a current graph $G_c$. Let $C_1$, $C_2$, \ldots, $C_{m-n+1}$ be the chains of the chain decomposition in the order of creation. We initialize $G_c$ to $C_1 \cup C_2$. In phase $i$, $i \in [1,m-n+1]$, we consider the $i$-th chain $C_i$ and either add all chains $C$ to $G_c$ for which the source vertex $s(C)$ s-belongs to $C_i$ to $G_c$ or exhibit a 2-edge-cut. As already mentioned, chains are added parent-first and hence $G_c$ is always parent-closed. We maintain the following invariant:\smallskip

\noindent{\textbf{Invariant:}} After phase $i$, $G_c$ consists of all chains for which the source vertex s-belongs to one of the chains $C_1$ to $C_i$.

\begin{lemma}\label{prop:ci_already_added}
For all $i$, the current chain $C_i$ is part of the current graph $G_c$ at the beginning of phase $i$ or the algorithm has exhibited a 2-edge-cut before phase $i$.
\end{lemma}
\begin{proof}
The initial current graph consists of chains $C_1$ and $C_2$ and hence the claim is true for the first and the second phase. Consider $i > 2$. The source vertex $s(C_i)$ s-belongs to a chain $C_j$ with $j < i$ (Lemma~\ref{facts about chain decomposition}.(\ref{s(C) s-belongs to earlier chain})) and hence $C_i$ is added in phase $j$. \qed
\end{proof}

The next lemma gives information about the chains for which the source vertex s-belongs to $C_i$. None of them belongs to $G_c$ at the beginning of phase $i$ (except for chain $C_2$ that belongs to $G_c$ at the beginning of phase $1$) and they form subtrees of the chain tree. Only the roots of these subtrees can be nested. All other chains are interlacing.

\begin{lemma}\label{lem:no_type_one} Assume that the algorithm reaches phase $i$ without exhibiting a 2-edge-cut. Let $C \not= C_2$ be a chain for which $s(C)$ s-belongs to $C_i$. Then $C$ is not part of $G_c$ at the beginning of phase $i$. Let $D$ be any ancestor of $C$ that is not in $G_c$. Then:
\begin{compactenum}[(1)]
\item $s(D)$ s-belongs to $C_i$.
\item If $D$ is nested, it is a child of $C_i$.
\item If $p(D)$ is not part of the current graph, $D$ is interlacing.
\end{compactenum}
\end{lemma}
\begin{proof}
We use induction on $i$. Consider the $i$-th phase and let $C \not= C_2$ be chains whose source vertex $s(C)$ s-belongs to $C_i$. We first prove that $C$ is not in $G_c$. This is obvious, since in the $j$-th phase we add exactly the chains whose source vertex $s$-belongs to $C_j$.

(1): Let $D$ be any ancestor of $C$ which is not part of $G_c$. By Lemma~\ref{facts about chain decomposition}, we have $s(D) \le s(C)$ and hence $s(D)$ belongs to $C_j$ for some $j\le i$. If $j < i$, $D$ would have been added in phase $j$, a contradiction to the assumption that $D$ does not belong to $G_c$ at the beginning of phase $i$.

(2): $s(D)$ s-belongs to $C_i$ by (1). If $D$ is nested, $s(D)$ and $t(D)$ s-belong to the same chain. Thus $D$ is a child of $C_i$.

(3): If $p(D)$ is not part of the current graph, $p(D)\neq C_i$ by Lemma~\ref{prop:ci_already_added} and hence $D$ is not a child of $C_i$. Hence by (2), $D$ is interlacing.\qed
\end{proof}

We can now define the segments with respect to $C_i$ by means of an equivalence relation. Consider the set $\cal S$ of chains whose source vertex s-belongs to $C_i$. For a chain $C \in \cal S$, let $C^*$ be the minimal ancestor of $C$ that does not belong to $G_c$. Two chains $C$ and $D$ in $\cal S$ belong to the same segment if and only if $C^* = D^*$. In Figure~\ref{fig:chains} on page~\pageref{fig:chains}, if we start with $G_c = C_1 \cup C_2$, we form three segments in the first phase, namely $\sset{C_4}$, $\sset{C_3}$, and $\sset{C_5}$. The first segment can be added according to Lemma~\ref{interlacing are easy}. Then $C_3$ can be added and then $C_5$.

Consider any $C \in \cal S$. By part (1) of the preceding lemma either $p(C) \in \cal S$ or $p(C)$ is part of $G_c$. Moreover, $C$ and $p(C)$ belong to the same segment in the first case. Thus segments correspond to subtrees in the chain tree. In any segment only the minimal chain can be nested by Lemma~\ref{lem:no_type_one}. If it is nested, it is a child of $C_i$ (parts (2) and (3) of the preceding lemma). Since only the root of a segment may be a nested chain, once it is added to the current graph all other chains in the segment can be added in parent-first order by Lemma~\ref{interlacing are easy}. All that remains is to find the proper ordering of the segments faster than in the previous section. We do so in Lemma~\ref{lem:add_hard_segments}. If no proper ordering exists, we exhibit a 2-edge-cut.



\begin{lemma}\label{cor:add_segment} All chains in a segment $S$ can be added in parent-first order if its minimal chain can be added.
\end{lemma}
\begin{proof} By Lemma~\ref{lem:no_type_one} all but the minimal chain in a segment are interlacing. Thus the claim follows from Lemma~\ref{interlacing are easy}. \qed
\end{proof}

We come to part I of phase $i$, the addition of all segments whose minimal chain is interlacing. As a byproduct, we will also determine all segments with nested minimal chain. We iterate over all chains $C$ whose source $s(C)$ s-belongs to $C_i$. For each such chain, we traverse the path $C$, $p(C)$, $p(p(C))$, \ldots until we reach a chain that belongs to $G_c$ or is already marked (initially, all chains are unmarked). We now distinguish cases. If the last chain on the path is nested we mark all chains on the path with the nested chain. If we hit a marked chain we copy the marker to all chains in the path.  Otherwise, i.e., all chains are interlacing and unmarked, we add all chains in the path to $G_c$ in parent-first order, as these segments can be added according to Corollary~\ref{cor:add_segment}. 

It remains to compute a proper ordering of the segments in which the minimal chain is nested or to exhibit a 2-edge-cut. We do so in part II of phase $i$. For simplicity, we will say `segment' instead of `segment with nested minimal chain' from now on.

For a segment $S$ let the \emph{attachment points} of $S$ be all vertices in $S$ that are in $G_c$. Note that the attachment points must necessarily be endpoints of chains in $S$ and hence adding the chains of $S$ makes the attachment points branch vertices. Nested children $C$ of $C_i$ can be added if there are branch vertices on $t(C)\rightarrow_T s(C)$, therefore adding a segment can make it possible to add further segments.

\begin{lemma}\label{prop:attachment points} Let $C$ be a nested child of $C_i$ and let $S$ be the segment containing $C$. The attachment points of $S$ consist of $s(C)$, $t(C)$, and the vertices $s(D)$ of the other chains in the segment. All such points lie on the path $t(C) \rightarrow_T s(C)$ and hence on $C_i$. 
\end{lemma} 
\begin{proof} Let $D$ be any chain in $S$ different from $C$. By Lemma~\ref{lem:no_type_one}, $C$ is the minimal chain in $S$. Since $S$ is a subtree of the chain tree, we have $C < D$ and hence by Lemma~\ref{facts about chain decomposition} $t(C) \le t(D)$. Since none of the chains in $S$ is part of $G_c$, parent-closedness implies that no vertex on the path $(t(D) \rightarrow_T t(C))\setminus t(C)$ belongs to $G_c$. In particular, either $t(D) = t(C)$ or $t(D)$ is not a vertex of $G_c$ and hence not an attachment point of $S$. It remains to show $s(C) \le s(D) \le t(C)$. Since $C \le D$, we have $s(C)\le s(D)$ by Lemma~\ref{facts about chain decomposition}. Since $s(D) \le t(D)$ and $t(C) \le t(D)$ we have either $s(D) \le t(C) \le t(D)$ or $t(C) < s(D) \le t(D)$. In the former case, we are done. In the latter case, $s(D)$ is not a vertex of $G_c$ by the preceding paragraph, a contradiction, since $s(D)$ s-belongs to $C_i$ by Lemma~\ref{lem:no_type_one}.
\qed
\end{proof}

For a set of segments $S_1,\ldots, S_k$, let the \emph{overlap graph} be the graph on the segments and a special vertex $R$ for the branch vertices on $C_i$. In the overlap graph, there is an edge between $R$ and a vertex $S_i$, if there are attachment points $a_1\le a_2$ of $S_i$ such that there is a branch vertex on the tree path $a_2\rightarrow_T a_1$. Further, between two vertices $S_i$ and $S_j$ there is an edge if there are attachment points $a_1$, $a_2$ in $S_i$ and $b_1$, $b_2$ in $S_j$, such that $a_1\le b_1\le a_2 \le b_2$ or $b_1 \le a_1 \le b_2 \le a_2$. We say that $S_i$ and $S_j$ \emph{overlap}.

\newcommand{\calC}{{\mathcal{C}}}

\begin{lemma}\label{overlap with calC suffices}
Let $\calC$ be a connected component of the overlap graph $H$ and let $S$ be any segment with respect to $C_i$ whose minimal chain $C$ is nested. Then $S\in \calC$ if and only if
  \begin{compactenum}[(i)]
  \item \label{connected by branch} $R\in \calC$ and there is a branch vertex on $t(C)\rightarrow_T s(C)$ or
  \item \label{connected by overlap} there are attachments $a_1$ and $a_2$ of $S$ and attachments $b_1$ and $b_2$ of segments in $\calC$ with $a_1 \le b_1 \le a_2 \le b_2$ or $b_1 \le a_1 \le b_2 \le a_2$.
  \end{compactenum}
\end{lemma}
\begin{proof} We first show $S\in \calC$ if (\ref{connected by branch}) or (\ref{connected by overlap}) holds. For (\ref{connected by branch}) the claim follows directly from the definition of the overlap graph. For (\ref{connected by overlap}), assume $S \not\in \calC$ for the sake of a contradiction. Then either $R\not \in \calC$ or there is no branch vertex in $t(C)\rightarrow_T s(C)$ by (\ref{connected by branch}). Further, no segment in $\calC$ overlaps with $S$ and hence any segment in $\calC$ has its attachments points either strictly between $a_1$ and $a_2$ or outside the path $a_2 \rightarrow_T a_1$. Moreover, both classes of segments are non-empty. However, segments in the two classes do not overlap and $R$ cannot be connected to the segments in the former class. Thus $\calC$ is not connected, a contradiction.

If neither (\ref{connected by branch}) nor (\ref{connected by overlap}) hold, there can be no segment in $\calC$ overlapping $S$ and either $S$ is not connected to $R$ or no segment in $\calC$ is connected to $R$.
\end{proof}

\begin{lemma}\label{lem:add_hard_segments}
  Assume the algorithm reaches phase $i$. If the overlap graph $H$ induced by the segments with respect to  $C_i$ is connected, we can add all segments of $C_i$. If $H$ is not connected, we can exhibit a 2-edge-cut for any component of $H$ that does not contain $R$.
\end{lemma}
\begin{proof} Assume first that $H$ is connected. Let $R,S_1,\ldots,S_k$ be the vertices of $H$ in a preorder, e.g.\ the order they are explored by a DFS, starting at $R$, the vertex corresponding to the branch vertices on $C_i$. An easy inductive argument shows that we can add all segments in this order. Namely, let $k \ge 1$ and let $C$ be the minimal chain of $S_k$. All attachment points of $S_k$ lie on the path $t(C) \rightarrow_T s(C)$ by Lemma~\ref{prop:attachment points}, and there is either an edge between $R$ and $S_k$ or an edge between $S_j$ and $S_k$ for some $j <k$.  In the former case, there is a branch vertex on $t(C) \rightarrow_T s(C)$ at the beginning of the phase, in the latter case there is one after adding $S_j$. Thus the minimal chain of $S_k$ can be added and then all other chains by Lemma~\ref{cor:add_segment}.

On the other hand, suppose $H$ is not connected. Let $\calC$ be any connected component of $H$ that does not contain $R$, and let $\calC_R$ be the connected component that contains $R$. Let $x$ and $y$ be the minimal and maximal attachment points of the segments in $\calC$, and let $G_c$ be the current graph after adding all chains in $\calC_R$. We first show that there is no branch vertex of $G_c$ on the path $y \rightarrow_T x$. Assume otherwise and let $w$ be any such branch vertex. Observe first that there must be a chain $C \in \calC$ with $s(C) \le w \le t(C)$. Otherwise, every chain in $\calC$ has all its attachment points at proper ancestors of $w$ or at proper descendants of $w$ and hence $\calC$ is not connected. Let $S$ be the segment containing $C$. By Lemma~\ref{prop:attachment points}, we may assume that $C$ is the minimal chain of $S$. Since $S \not\in \calC_R$, $\edge{R}{S}$ is not an edge of $H$ and hence no branch vertex exists on the path $t(C)  \rightarrow_T s(C)$ at the beginning of part II of the phase. Hence $w$ is an attachment point of a segment in $\calC_R$. In particular $\calC_R$ contains at least one segment. We claim that $\calC_R$ must also have an attachment point outside $t(C)  \rightarrow_T s(C)$. This holds since all initial branch vertices are outside the path and since $\calC_R$ is connected. Thus $S \in \calC_R$ by Lemma~\ref{overlap with calC suffices},  a contradiction.

We show next that the tree-edge $\edge{x}{p(x)}$ and the edge $\edge{z}{y}$ from $y$'s predecessor $z$ on $C_i$ to $y$ form a 2-edge-cut; $\edge{z}{y}$ may be a tree-edge or a back-edge. The following argument is similar to the argument in Theorem~\ref{order exists}, but more refined.

Assume otherwise. Then, as in the proof of Theorem~\ref{order exists}, there is a path $P=a\rightarrow b$ such that $a\le u$ for all $u\in P$, and either $a$ lies on $y\rightarrow_T x$ and $b$ does not, or vice versa, and no inner vertex of $P$ is in $G_c$. Moreover, the first edge $\edge av$ of $P$ is a back-edge and $v$ is a descendant of $b$. Note that unlike in the proof of  Theorem~\ref{order exists}, $a$ and $b$ need not lie on different links, as we want to show that $\edge{x}{p(x)}$ and $\edge{z}{y}$ form a cut and these might be different from the last edges on the link containing $x$ and $y$.

Let $D$ be the chain that starts with the edge $\edge{a}{v}$. $D$ does not belong to $G_c$, as no edge of $P$ belongs to $G_c$. In particular, $a$ does not s-belong to $C_j$ for $j < i$ (as otherwise, $D$ would already be added). Since $a \le b$ and one of $a$ and $b$ lies on $y \rightarrow_T x$ (which is a subpath of $C_i$), $a$ s-belongs to $C_i$. By the argument from the proof of Theorem~\ref{order exists}, $t(D)$ is a descendant of $b$.

Let $D^*$ be the chain that contains the last edge of $P$. If $t(D) = b$, $D = D^*$. Otherwise, $t(D)$ is a proper descendant of $b$. Let $\edge{u}{b}$ be the last edge on the path $t(D) \rightarrow_T b$.  We claim that $\edge{u}{b}$ is also the last edge of $P$. This holds since the last edge of $P$ must come from a descendant of $b$ (as ancestors of $b$ belong to $G_c$) and since it cannot come from a child different from $y$ as otherwise $P$ would have to contains a cross-edge. Thus $D^* \le D$ by Lemma~\ref{facts about chain decomposition}.(\ref{t(D) above u}) and hence $s(D^*) \le s(D) \le a$ by part (\ref{parent on chains follows parent on nodes}) of the same lemma.

$D$ and $D^*$ belong to the same segment with respect to $C_i$, say $S$, and $a$ and $b$ are vertices in $S \cap G_c$. This can be seen easily. Since $a$ s-belongs to $C_i$, $D$ belongs to some segment with respect to $C_i$ and since $D^* \le D$, $D^*$ belongs to the same segment. Since $t(D^*) = b$ and $b$ is a vertex of $G_c$, $D^*$ is the minimal chain in $S$. Thus $D^*$ is nested and hence $b$ s-belongs to $C_i$. Hence $a$ and $b$ are attachment points of $S$.

Thus $S$ overlaps with $\calC$ and hence $S \in \calC$ by Lemma~\ref{overlap with calC suffices}. Therefore $x$ and $y$ are not the extremal attachment points, that is the minimal (or maximal) vertices in $S\cap G_c$, of $\calC$, a contradiction. \qed
\end{proof}

It remains to show that we can find an order as required in Lemma~\ref{lem:add_hard_segments}, or a 2-edge-cut, in linear time. We reduce the problem of finding an order on the segments to a problem on intervals. W.l.o.g.\ assume that the vertices of $C_i$ are numbered consecutively from $1$ to $|C_i|$. Consider any segment $S$, and let $a_0\leq a_1\leq \ldots \leq a_k$ be the set of attachment points of $S$, i.e., the set of vertices that $S$ has in common with $C_i$. By Lemma~\ref{prop:attachment points}, $a_0$ and $a_k$ are the endpoints of the minimal chain in $S$ and each $a_i$, $0 < i < k$, is equal to $s(D)$ for some other chain in $S$. We associate the intervals%
\[
  \{[a_0,a_\ell] | 1\leq \ell \leq k\} \cup \{[a_\ell,a_k] | 1\leq \ell< k\},
\]
 with $S$ and for every branch vertex $v$ on $C_i$ we define an interval $[0,v]$. See Figure~\ref{fig:attachment_point_intervals} for an example.
\begin{figure}
\centering
\definecolor{cffffff}{RGB}{255,255,255}
\definecolor{c0000ff}{RGB}{0,0,255}
\begin{tikzpicture}[y=0.60pt, x=0.6pt,yscale=-1, inner sep=0pt, outer sep=0pt]
    \path[shift={(0,1.66744)},fill=black]
      (134.8939,81.4960)arc(0.000:180.000:4.000)arc(-180.000:0.000:4.000) -- cycle;
    \path[draw=black,dashed,line join=miter,line cap=butt,line width=0.800pt]
      (159.4488,83.1635) -- (302.2232,83.1635) .. controls (302.2232,83.1635) and
      (248.3759,41.0607) .. (216.5586,41.0607) .. controls (184.7413,41.0607) and
      (130.8939,83.1635) .. (130.8939,83.1635);
    \path[fill=black] (154.84824,98.587296) node[above right] (text3027) {$1$};
    \path[fill=black] (202.43971,98.587296) node[above right] (text3031) {$2$};
    \path[fill=black] (250.03119,98.587296) node[above right] (text3035) {$3$};
    \path[fill=black] (297.62265,98.587296) node[above right] (text3039) {$4$};
    \path[shift={(2.01483,-1.25057)},draw=black,fill=cffffff]
      (256.6169,84.4141)arc(0.000:180.000:4.000)arc(-180.000:0.000:4.000) -- cycle;
    \path[fill=black] (416.13892,98.587296) node[above right] (text3027-8) {$1$};
    \path[fill=black] (458.60425,98.587296) node[above right] (text3031-6) {$2$};
    \path[fill=black] (500.96216,98.587296) node[above right] (text3035-9) {$3$};
    \path[fill=black] (543.7937,98.587296) node[above right] (text3039-1) {$4$};
    \path[fill=black] (373.78101,98.728889) node[above right] (text4128) {$0$};
    \path[draw=black,line join=miter,line cap=butt,line width=0.800pt]
      (374.9402,41.6045) -- (417.2375,41.6045);
    \path[draw=black,line join=miter,line cap=butt,line width=0.800pt]
      (417.2375,39.6244) -- (417.2375,43.5846);
    \path[draw=black,line join=miter,line cap=butt,line width=0.800pt]
      (374.9402,39.6244) -- (374.9402,43.5846);
    \path[draw=black,line join=miter,line cap=butt,line width=0.800pt]
      (374.9402,39.6244) -- (374.9402,43.5846);
    \path[draw=c0000ff,line join=miter,line cap=butt,line width=0.800pt]
      (159.4488,83.1635) .. controls (163.8978,75.9450) and (248.9835,20.3629) ..
      (285.2178,38.3891) .. controls (295.8598,60.2515) and (302.2232,83.1635) ..
      (302.2232,83.1635);
    \path[draw=c0000ff,line join=miter,line cap=butt,line width=0.800pt]
      (207.0403,83.1635) .. controls (232.2578,58.7003) and (258.3328,43.9550) ..
      (285.2178,38.3891);
    \path[shift={(15.70172,1.25058)},draw=black,fill=cffffff]
      (195.3386,81.9129)arc(0.000:180.000:4.000)arc(-180.000:0.000:4.000) -- cycle;
    \path[shift={(32.86245,-45.85455)},draw=black,fill=cffffff]
      (256.6169,84.4141)arc(0.000:180.000:4.000)arc(-180.000:0.000:4.000) -- cycle;
    \path[shift={(-1.04214,1.25058)},fill=black]
      (164.4910,81.9129)arc(0.000:180.000:4.000)arc(-180.000:0.000:4.000) -- cycle;
    \path[draw=black,fill=cffffff]
      (306.2232,83.1635)arc(0.000:180.000:4.000)arc(-180.000:0.000:4.000) -- cycle;
    \path[draw=black,line join=miter,line cap=butt,line width=0.800pt]
      (417.2376,81.8313) -- (544.2820,81.8313);
    \path[draw=black,line join=miter,line cap=butt,line width=0.800pt]
      (544.2820,79.8512) -- (544.2820,83.8114);
    \path[draw=black,line join=miter,line cap=butt,line width=0.800pt]
      (417.2376,79.8512) -- (417.2376,83.8114);
    \path[draw=black,line join=miter,line cap=butt,line width=0.800pt]
      (417.2376,79.8512) -- (417.2376,83.8114);
    \path[draw=black,line join=miter,line cap=butt,line width=0.800pt]
      (417.2376,68.4224) -- (459.3367,68.4224);
    \path[draw=black,line join=miter,line cap=butt,line width=0.800pt]
      (459.3367,66.4423) -- (459.3367,70.4025);
    \path[draw=black,line join=miter,line cap=butt,line width=0.800pt]
      (417.2376,66.4423) -- (417.2376,70.4025);
    \path[draw=black,line join=miter,line cap=butt,line width=0.800pt]
      (417.2376,66.4423) -- (417.2376,70.4025);
    \path[draw=black,line join=miter,line cap=butt,line width=0.800pt]
      (459.6248,55.0134) -- (544.2820,55.0134);
    \path[draw=black,line join=miter,line cap=butt,line width=0.800pt]
      (544.2820,53.0334) -- (544.2820,56.9935);
    \path[draw=black,line join=miter,line cap=butt,line width=0.800pt]
      (459.6248,53.0334) -- (459.6248,56.9935);
    \path[draw=black,line join=miter,line cap=butt,line width=0.800pt]
      (459.6248,53.0334) -- (459.6248,56.9935);
\end{tikzpicture}
\caption{Intervals for the solid segment with respect to the dashed chain. It has the attachment points $1$,$2$,$4$. Filled vertices are branching points.}
\label{fig:attachment_point_intervals}
\end{figure}
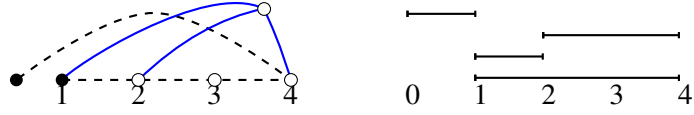

We say two intervals $[a, a'], [b, b']$ \emph{overlap} if $a\leq b\leq a'\leq b'$. Note that overlapping is different from intersecting; an interval does not overlap intervals in which it is properly contained or which it properly contains. This relation naturally induces a graph $H'$ on the intervals. Contracting all intervals that are associated to the same segment into one vertex makes $H'$ isomorphic to the overlap graph as required for Lemma~\ref{lem:add_hard_segments}. Hence we can use $H'$ to find the order on the segments. Note that the interval set $\{[a_0,a_\ell] | 1\leq \ell \leq k\}$ for each segment does not suffice without employing a clever tie-breaking rule: If there are two segments with attachments $a<b<c$ and $a<b'<c$, respectively, such that $b' \neq b$, no interval of the first segment overlaps with one of the second.

A naive approach that constructs $H'$, contracts intervals, and runs a DFS will fail, since the overlap graph can have a quadratic number of edges. However, using a method developed by Olariu and Zomaya~\cite{Olariu1996}, we can compute a spanning forest of $H'$ in time linear in the number of intervals. The presentation in~\cite{Olariu1996} is for the PRAM and thus needlessly complicated for our purposes. A simpler explanation can be found in the appendix.

The number of intervals created for a chain $C_i$ is bounded by
\[
	|\text{NestedChildren}(C_i)|+2|\text{Interlacing}(C_i)|+|V_{\text{branch}}(C_i)|,
\]
where $\text{NestedChildren}(C_i)$ are the nested children of $C_i$, $\text{Interlacing}(C_i)$ are the interlacing chains that start on $C_i$, and $V_{\text{branch}}(C_i)$ is the set of branch vertices on $C_i$. Note that we generate 
the interval $[s(C),t(C)]$ for each nested child $C$, and the intervals $[s(C),s(D)]$ and $[s(D),t(C)]$ for each interlacing chain $D$ belonging to a segment with nested minimal chain $C$. Thus the total time spend the ordering procedure is $O(m)$. 
From the above discussion, we get:

\begin{theorem}
For a $3$-edge-connected graph, a Mader construction sequence can be found in time O($n+m$).
\end{theorem}

\section{Verifying the Mader Sequence}\label{sec:verify mader}

The certificate is either a 2-edge-cut, or a sequence of Mader-paths. For a 2-edge-cut, we simply remove the two edges and verify that $G$ is no longer connected.

For checking the Mader sequence, we doubly-link each edge in a Mader-path to the corresponding edge in $G$. Let $G'$ be a copy of $G$. We remove the Mader-paths, in reverse order of the sequence, suppressing vertices of degree two as they occur. This can create multiple edges and loops. Let $G'_i$ be the multi-graph before we remove the $i$-th path $P_i$. We need to verify the following:

\begin{compactitem}
\item $G$ must have minimum degree three.
\item The union of Mader-paths must be isomorphic to $G$ and the Mader-paths must partition the edges of $G$. This is easy to check using the links between the edges of the paths and the edges of $G$.
\item The paths we remove must be ears. More precisely, at step $i$, $P_i$ must have been reduced to a single edge in $G'_i$, as inner vertices of $P_i$ must have been suppressed if $P_i$ is an ear for $G'_i$.
\item The $P_i$ must not subdivide the same link twice. That is, after deleting the edge corresponding to $P_i$, it must not be the case that both endpoints are still adjacent (or equal, i.e.\ $P_i$ is a loop) but have degree two.
\item When only two paths are left, the graph must be a $K_2^3$.
\end{compactitem}


\section{The Cactus Representation of 2-Cuts}\label{Sec: Cactus Representation}
We review the cactus representation of 2-cuts in a 2-connected but not 3-connected graph and show how to certify it.

A \emph{cactus} is a graph in which every edge is contained in exactly one cycle. Dinits, Karzanov, and Lomonsov~\cite{Dinits-Karzanov-Lomonosov} showed that the set of mincuts of any graph has a cactus representation, i.e., for any graph $G$ there is a cactus $C$ and a mapping $\phi: V(G) \rightarrow V(C)$ such that the mincuts of $G$ are exactly the preimages of the mincuts of $C$, i.e., for every mincut\footnote{For this theorem, a cut is specified by a set of vertices, and the edges in the cut are the edges with exactly one endpoint in the vertex set.} $A \subseteq V(C)$, $\phi^{-1}(A)$ is a mincut of $G$, and all mincuts of $G$ can be obtained in this way. The pair $(C,\phi)$ is called a \emph{cactus representation} of $G$. Fleiner and Frank~\cite{Fleiner2009} provide a simplified proof for the existence of a cactus representation. We will call the elements of $V(G)$ \emph{vertices}, the elements of $V(C)$ \emph{nodes}, and the preimages of nodes of $C$ \emph{blobs}.

In general, a cactus representation needs to include nodes with empty preimages. This happens for example for the $K_4$; its cactus is a star with double edges where the central node has an empty preimage and the remaining nodes correspond to the vertices of the $K_4$. For graphs whose mincuts have size two, nodes with empty preimages are not needed, and a cactus representation can be obtained by contracting the 3-edge-connected components into a single node. 

\newcommand{\citeNaga}{\cite[Section 2.3.5]{Nagamochi-Ibaraki-Book}}

\begin{lemma}[\citeNaga] \label{cactus representation} Let $G$ be a 2-edge-connected graph that is not 3-edge-connected. Contracting each 3-edge-connected components of $G$ into a node yields a cactus representation $(C,\phi)$ of $G$ with the following properties:
\begin{compactenum}[i)]
\item The edges of $C$ are in one-to-one correspondence to the edges of $G$ that are contained in a 2-cut.
\item For every node $c\in V(c)$, $\phi^{-1}(c)$ is a 3-edge-connected component of $G$.
\end{compactenum}
\end{lemma}


\subsection{Verifying a Cactus Representation}

Let $G$ be a graph and let $(C,\phi)$ be an alleged cactus-representation of its 2-cuts in the sense of Lemma~\ref{cactus representation}. We show how to verify a cactus representation in linear time. We need to check two things. First, we need to ensure that $C$ is indeed a cactus graph, that is, every edge of $C$ is contained in exactly one cycle, that $\phi$ is a surjective mapping and hence there are no empty blobs, and that every edge of $G$ either connects two vertices in the same blob or is also present in $C$. Second, we need to verify that the blobs of $C$ are $3$-edge-connected components of $G$. For this purpose, the cactus representation is augmented by a Mader construction sequence for each blob $B$. The verification procedure from Sect.~\ref{sec:verify mader} can then be applied.

We first verify that $C$ is a cactus. We compute a chain decomposition of $C$ and verify that every chain is a cycle. We label all edges in the $i$-th cycle by $i$. We have now verified that $C$ is a cactus.

Surjectivity of $\phi$ is easy to check. We then iterate over the edges $\edge uv$ of $G$. If its endpoints belong to the same blob, we associate the edge with the blob. If its endpoints do not belong to the same blob, we add the pair $\edge{\phi(u)}{\phi(v)}$ to a list. Having processed all edges, we check whether the constructed list and the edge list of $C$ are identical by first sorting both lists using radix sort and then comparing them for identity.

We finally have to check that the blobs of $C$ correspond to 3-edge-connected components of $G$. Our goal is to use the certifying algorithm for $3$-edge-connectivity on the substructures of $G$ that represent $3$-edge-connected components. Let $B$ be any blob. We already collected the edges having both endpoints in $B$. We also have to account for the paths using edges outside $B$.  We do so by creating an edge $\edge uv$ for a every path in $G$ leaving $B$ at vertex $u$ and returning to $B$ at vertex $v$. 
It is straightforward to compute these edges; we look at all edges having exactly one endpoint in the blob. Each such edge corresponds to an edge in $C$. For each such edge, we know to which cycle it belongs. The outgoing edges pair up so that the two edges of each pair belong to the same cycle. 

The maximality of each blob $B$ is given by the fact that every edge of $C$ is contained in a 2-edge-cut of $C$ and hence contained in a 2-edge-cut of $G$.

Every algorithm for computing the 3-edge-connected components of a graph, e.g.~\cite{Nagamochi1992a,Taoka1992,Tsin2007,Tsin2009,Nagamochi-Ibaraki-Book}, can be turned into a certifying algorithm for computing the cactus representation of 2-cuts. We obtain the cactus $C$ and the mapping $\phi$ by contraction of the 3-edge-connected components (Lemma~\ref{cactus representation}). Then one applies our certifying algorithm for 3-edge-connectivity to each 3-edge-connected component. The drawback of this approach is that it requires  \emph{two} algorithms that check 3-connectivity. 
In the next section we will show how to extend our algorithm so that it computes the 3-edge-connected components and the cactus representation of 2-cuts of a graph directly.

\section{Computing a Cactus Representation}\label{Computing a Cactus}

We discuss how to extend the algorithm to construct a cactus representation. We begin by examining the structure of the 2-cuts of $G$ more closely to extend our algorithm such that it finds all 2-cuts of the graph and encodes them efficiently.

We will first show that the two edges of every $2$-edge-cut of $G$ are contained in a common chain. This restriction allows us to focus on the $2$-edge-cuts that are contained in the currently processed chain $C_i$ only. 
In the subsequent section, we show how to maintain a cactus for every phase $i$ of the algorithm that represents all $2$-edge-cuts of the graph of the branch vertices and links of $C_1 \cup \ldots \cup C_i$ in linear space. The final cactus will therefore represent all $2$-edge-cuts in $G$.

There is one technical detail regarding the computation of overlap graphs: For the computation of a Mader-sequence in Section~\ref{sec:linear time alg}, we stopped the algorithm when the first $2$-edge-cut occurred, as then a Mader-sequence does not exist anymore. Here, we simply continue the algorithm with processing the next chain $C_{i+1}$. This does not harm the search for cuts in subsequent chains, as the fact that $2$-edge-cuts are only contained in common chains guarantees that every $2$-edge-cut that contains an edge $e$ in $C_i$ has its second edge also in $C_i$.

For simplicity, we assume that $G$ is 2-edge-connected and has minimum degree three from now on. Then all 3-edge-connected components contain at least two vertices.

\subsection{2-Edge-Cuts are Contained in Chains and an Efficient Representation of All Cuts in a Chain}\label{sec:efficient cut representation}

In phase $i$ of the algorithm, using Lemma~\ref{lem:add_hard_segments}, we can find a 2-edge-cut for each connected component of the overlap graph $H$ that does not contain $R$ ($R$ is the special vertex in $H$ that represents the branch vertices on $C_i$). Lemma~\ref{lem:h induces all cuts} shows that the set of edges contained in these cuts is equal to the set of edges contained in any cut on $C_i$.
Lemma~\ref{lem:facts about cuts} states easy facts about $2$-edge-cuts, in particular, that the edges of any 2-edge-cut are contained in a common chain. The proofs can be found in many 3-connectivity papers, e.g.~\cite{Nagamochi1992a,Taoka1992,Tsin2007,Tsin2009}. As in the previous sections, all DFS-tree-edges are oriented towards the root, while back-edges are oriented away from the root.

\begin{lemma}\label{lem:facts about cuts}
Let $T$ be a DFS-tree of a $2$-edge-connected graph $G$. Every $2$-edge-cut $(\edge uv,\edge xy)$ of $G$ satisfies the following:

\begin{compactenum}[(1)]
  \item At least one of $\edge uv$ and $\edge xy$ is a tree-edge, say $\edge xy$.\label{not both backedges}
  \item $G-\edge uv-\edge xy$ has exactly two components. Moreover, the edges $\edge uv$ and $\edge xy$ have exactly one endpoint in each component.\label{endpoints disconnected}
  \item The vertices $u$, $v$, $x$, and $y$ are contained in the same leaf-to-root path of $T$.\label{on a tree path}
  \item If $\edge uv$ and $\edge xy$ are tree-edges and w.l.o.g.\ $u \leq y$, the vertices in $y \rightarrow_T u$ and $\{x,v\}$ are in different components of $G-\edge uv-\edge xy$.\label{components two tree-edges}
  \item If $\edge uv$ is a back-edge, then $\edge xy \in (v \rightarrow_T u)$ and, additionally, the vertices in $v \rightarrow_T x$ and $y \rightarrow_T u$ are in different components of $G-\edge uv-\edge xy$.\label{components back and tree-edge}
\end{compactenum}\smallskip
\noindent Moreover, let $\calC$ be a chain decomposition of $G$. For every $2$-edge-cut $\{\edge uv,\edge xy\}$ of $G$, $\edge uv$ and $\edge xy$ are contained in a common chain $C \in \calC$.

\end{lemma}

\begin{lemma}\label{lem:h induces all cuts}
Let $\mathcal E$ be the set of edges that are contained in the 2-edge cuts induced by the connected components of the overlap graph $H$ at the beginning of part II of phase $i$. Then any 2-edge-cut $\{\edge xy, \edge uv\}$ on $C_i$ is a subset of $\mathcal E$.
\end{lemma}
\begin{proof}
Assume for the sake of contradiction that there is an edge $\edge uv$ in the $2$-edge-cut that is not in $\mathcal E$. We distinguish the following cases.

First assume that both $\edge uv$ and $\edge xy$ are tree-edges and w.l.o.g.\ $v<u\le y<x$. Since $G$ has minimal degree three, every vertex on $C_i$ has an incident edge that is not on $C_i$. Hence it is either a branch vertex, or belongs to some segment with respect to $C_i$ (incident back-edges start chains in segments w.r.t.\ $C_i$, incident tree edges are the last edges of chains in segments w.r.t.\ $C_i$). As $s(C_i) \le v$ is a branch vertex, by Lemma~\ref{lem:facts about cuts}.(\ref{components two tree-edges}) the path $y\rightarrow_T u$ can not contain a branch vertex. In particular, $u$ is not a branch vertex.

Let $S_u$ be any segment having $u$ as attachment vertex. All segments in the connected component of $S_u$ in $H$ must have their attachment vertices on $y\rightarrow _T u$ and the connected component does not contain $R$. Hence this connected component induces a cut containing $\edge uv$.

Now assume that one of $\edge uv$ and $\edge xy$ is a back-edge. If $\edge uv$ is the back-edge, then $u=s(C_i)$ and we have $u<y<x<v$ by Lemma~\ref{lem:facts about cuts}. The path $v\rightarrow_T x$ cannot contain a branch vertex. Let $S_v$ be any segment that has $v$ as attachment vertex. All segments in the connected component of $S_v$ must have their attachment vertices on $v\rightarrow_T x$ and the connected component does not contain $R$. Hence $\edge uv$ is contained in a cut induced by this connected component.

If on the other hand $\edge uv$ is the tree-edge we have $y<v<u<x$ basically the same argument applies when we replace $S_v$ by a segment $S_u$ containing $u$.
\end{proof}


We next show how to compute a space efficient representation of all 2-cuts on the chain $C_i$. Using this technique we can store all 2-cuts in $G$ in linear space. In the next section we will then use this to construct the cactus-representation of all 2-cuts in $G$.

Number the edges in $C_i$ as $e_1$, $e_2$, \ldots, $e_k$. Here $e_1$ is a back edge and $e_2$ to $e_k$ are tree edges. We start with a simple observation. Let $h < i < j$. If $(e_h,e_i)$ and $(e_i,e_j)$ are 2-edge-cuts, then $(e_i,e_j)$ is a 2-edge-cut.

Using this observation, we want to group the edges of $2$-edge-cuts of $C_i$ such that (i) every two edges in a group form a $2$-edge-cut and (ii) no two edges of different groups form a $2$-edge-cut. The existence of such a grouping has already been observed in~\cite{Nagamochi1992a,Taoka1992,Tsin2009}. We show how to find it using the data structures we have on hand during the execution of our algorithm.

Consider the overlap graph $H$ in phase $i$ of our algorithm. We need some notation. Let $I$ be the set of intervals on $C_i$ that contains for every component of $H$ (except the component representing the branch vertices on $C_i$) with extremal attachment vertices $a$ and $b$ the interval $[a,b]$. Since the connected components of $H$ are maximal sets of overlapping intervals, $I$ is a \emph{laminar family}, i.e.\ every two intervals in $I$ are either disjoint or properly contained in each other. In particular, no two intervals in $I$ share an endpoint.
The layers of this laminar family encode which edges form pairwise 2-cuts in $G$, see Figure~\ref{fig:laminar family}. We define an equivalence relation to capture this intuition.

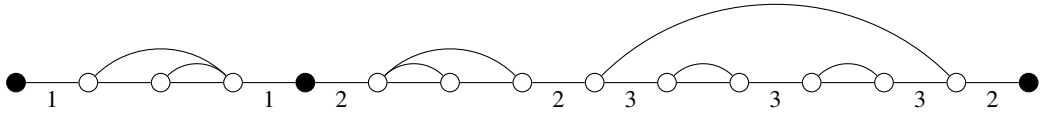
\begin{figure}
\centering
\begin{tikzpicture}[node distance=4ex]
\tikzstyle{vertex}=[circle, draw, inner sep=0.5ex]
\node[vertex, fill=black] (1) {};
\foreach \i in {2,3,4} {
  \pgfmathtruncatemacro{\left}{\i-1}
  \node[vertex, right=of \left] (\i) {};
}
\node[vertex, fill=black, right=of 4] (5) {};
\foreach \i in {6,...,14} {
    \pgfmathtruncatemacro{\left}{\i-1}
  \node[vertex, right=of \left] (\i) {};
}
\node[vertex, fill=black, right=of 14] (15) {};

\foreach \i/\j in {2/4, 3/4, 6/7, 6/8, 9/14,10/11,12/13} {
  \draw (\i) edge [out=45, in=135] (\j);
}
\foreach \i in {3,4,7,8,11,13} {
  \pgfmathtruncatemacro{\left}{\i-1}
  \draw (\i) edge (\left);
}

\footnotesize
\foreach \i in {2,5} {
  \pgfmathtruncatemacro{\left}{\i-1}
  \draw (\i) edge node [below] {1} (\left);
}
\foreach \i in {6,9,15} {
  \pgfmathtruncatemacro{\left}{\i-1}
  \draw (\i) edge node [below] {2} (\left);
}
\foreach \i in {10,12,14} {
  \pgfmathtruncatemacro{\left}{\i-1}
  \draw (\i) edge node [below] {3} (\left);
}
\end{tikzpicture}
\caption{The intervals induced by the connected components of the overlap graph $H$ form a laminar family. The levels of this family encode which edges form pairwise 2-cuts. Two edges in the figure are labeled with the same number if they form a cut. Filled vertices are branch vertices.}
\label{fig:laminar family}
\end{figure}

For an interval $[a,b]$, $a < b$, let $\ell([a,b])$ and $r([a,b])$ be the edges of $C_i$ directly before and after $a$ and $b$, respectively. We call $\{\ell([a,b]),r([a,b])\}$ the \emph{interval-cut} of $[a,b]$. For a subset $S \subseteq I$ of intervals, let $E_S$ be the union of edges that are contained in interval-cuts of intervals in $S$. According to Lemma~\ref{lem:h induces all cuts}, every $2$-edge-cut in $C_i$ consists of edges in $E_I$.

We now group the edges of $E_I$ using the observation above. Let two intervals $I_1 \in I$ and $I_2 \in I$ \emph{contact} if $r(I_1) = \ell(I_2)$ or $\ell(I_1) = r(I_2)$. Clearly, the transitive closure $\equiv$ of the contact relation is an equivalence relation. 
Every block $B$ of $\equiv$ is a set of pairwise disjoint intervals which are contacting consecutively.
This allows us to compute the blocks of $\equiv$ efficiently. We can compute them in time $|I|$ and store them in space $|I|$ by using a greedy algorithm that iteratively extracts the inclusion-wise maximal intervals in $I$ that are contacting consecutively. 

\begin{lemma}[\cite{Nagamochi1992a,Taoka1992,Tsin2009}]
Two edges $e$ and $e'$ in $C_i$ form a $2$-edge-cut if and only if $e$ and $e'$ are both contained in $E_B$ for some block $B$ of $\equiv$.
\end{lemma}

\subsection{An Incremental Cactus Construction}

In this section we show how to construct a cactus representation incrementally along our algorithm for constructing a Mader sequence. At the beginning of each phase $i$, we will have a cactus for the graph $G^i$ whose vertices are the branch vertices that exist at this time and whose edges are the links between these branch vertices.

We assume that $G$ is 2-edge-connected but not 3-edge connected and that $G$ has minimum degree three. This ensures that in phase $i$ every vertex on the current chain $C_i$ belongs to some segment or is a branch vertex.

We will maintain a cactus representation $(C,\phi)$, i.e., for every node $v$ of $C$, the blob $B = \phi^{-1}(v)$ is the vertex-set of a $3$-edge-connected component in $G^i$. We begin with a single blob that consists of the two branch vertices of the initial $K_2^3$, which clearly are connected by three edge-disjoint paths.

Consider phase $i$, in which we add all chains whose source s-belongs to $C_i$. At the beginning of the phase, the endpoints of $C_i$ and some branch vertices on $C_i$ already exist in $G^i$. We have a cactus representation of the current graph. The endpoints of $C_i$ are branch vertices and belong to the same blob $B$, since 2-edge-cuts are contained in chains.

We add all segments that do not induce cut edges and tentatively assign all vertices of $C_i$ to $B$. If the algorithm determines that $C_i$ does not contain any $2$-edge-cut, the assignment becomes permanent, the phase is over and we proceed to phase $i+1$. Otherwise we calculate the efficient representation of $2$-edge-cuts on $C_i$ from Sect.~\ref{sec:efficient cut representation}.

Let $e_1$ be the first edge on $C_i$ in a $2$-edge-cut, let $A$ be a block of the contact equivalence relation described in the last section containing $e_1$ and let $E_A = \sset{e_1,e_2,\ldots,e_\ell}$ such that $e_j$ comes before $e_{j+1}$ in $C_i$ for all $j$. Then every two edges in $E_A$ form a 2-edge-cut. We add a cycle with $\ell-1$ empty blobs $B_2,\ldots,B_\ell$ to $B$ in $C$. The $\ell$ new edges correspond to the $\ell$ edges in $E_A$.

For every pair $e_j=(a,b)$, $e_{j+1}=(c,d)$ in $E_A$ we remove the vertices between these edges from $B$. Since the edges in $C_i$ are linearly ordered, removing the vertices in a subpath takes constant time. We place the end vertices $b$ and $c$ of the path between $e_j$ and $e_{j+1}$ in the blob $B_j$, add the segments that induced this cut and recurse on the path between $b$ and $c$. That is, we add all vertices on the path from $b$ to $c$ to $B_j$, check for cut edges on this path, and, should some exist, add more blobs to the cactus. The construction takes constant time per blob. Figure~\ref{fig:cactus} shows an example.

\begin{figure}[t]
\centering
\begin{tikzpicture}
\tikzstyle{vertex}=[circle, draw, inner sep=0.5ex]
\footnotesize
\node (ci) {\normalsize $C_i$:};
\node [vertex, right=of ci] (0) {0};
\foreach \n/\e in {1/a,2/b,3/c,4/d,5/e,6/f,7/g,8/h} {
  \pgfmathtruncatemacro{\left}{\n-1}
  \node [vertex] (\n) [right=of \left] {\n};
  \draw  (\left) -- node[below=5mm, anchor=base] {\e} (\n);
};

\node (x1) [above= of 4,vertex] {};
\draw (1) edge [in=180] (x1);
\draw (x1) edge (4)
  edge [bend left] (5);
\draw (2) edge [out =45] (3);

\draw (6) edge [out = 45] (7);
\end{tikzpicture}\\\vspace{5ex}
\begin{tikzpicture}
\footnotesize
\tikzstyle{blob}=[draw, ellipse]
\node (a) [blob] {0, 8};
\node (b) [blob, above right=of a] {6, 7};
\node (c) [blob, above left = of a] {1, 4, 5};
\node (d) [blob, above left = of c] {2, 3};

\draw (a) edge node [right] {h} (b)
  edge node [left] {a} (c);
\draw (c) edge [bend left] node [left] {b} (d)
  edge [bend right] node [right] {d} (d);
\draw (b) edge node [above] {f} (c);
\end{tikzpicture}
\caption{\label{fig:cactus} The segments attached to chain $C_i$ and the corresponding part of the cactus. We first tentatively assign vertices 1--7 to the blob containing the endpoints $\{0,8\}$ of $C_i$. The top level cuts  are the pairs in the block $\sset{a,f,h}$. So we create a cycle with three edges and attach it to the blob containing 0 and 8. We move vertices 1--5 to the blob between $a$ and $f$, vertices 6--7 to the blob between $f$ and $h$, and keep vertices $0$ and $8$ in the parent blob. We then recurse into the first blob. The second level cuts are the pairs in the block $\sset{b,d}$. So we create a cycle with two edges and move vertices 2 and 3 to the new blob.}
\end{figure}
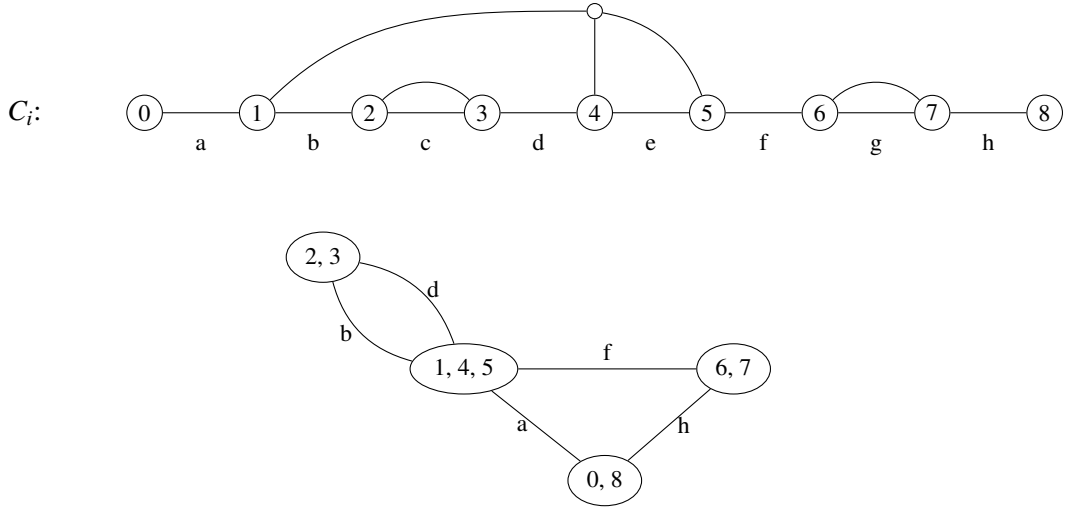

Graphs that contain nodes of degree two can be handled in the same way, if we add a cycle to each degree two node $u$. This cycle creates a segment w.r.t.\ the chain to which $u$ s-belongs and hence the algorithm correctly identifies the two incident edges as cut edges.



\begin{lemma} The above incremental procedure constructs a cactus representation of the 2-edge-cuts in $G$ in linear time.
\end{lemma}
\begin{proof} Each vertex in $G$ s-belongs to some chain. In the phase in which that chain is treated, all its vertices are added to a blob. Whenever we move a vertex to different blob, we remove it from its previous blob. Therefore each vertex of $G$ is contained in exactly one blob.

Whenever we add edges to the cactus, we do so by adding a cycle that shares exactly one node with the existing cactus. Hence every edge in the cactus lies on exactly one cycle.

Let $\{e_1,e_2\}$ be a 2-edge-cut in $G$. The two edges must lie on a common cycle in the cactus, since the edges in the cactus are in one-to-one correspondence with edges of $G$ and cutting a cycle in only one place cannot disconnect a graph. As the cycles of the cactus touch in at most one vertex, $e_1$ and $e_2$ are a cut in the cactus as well.

Conversely let $e_1'$, $e_2'$ be a cut in the cactus and let $e_1$, $e_2$ be the corresponding edges in $G$. Then $e_1'$, $e_2'$ must lie on some common cycle which, upon their removal, is split into two nonempty parts $H_1$, and $H_2$. Assume that $G -e_1 -e_2$ is still connected, then there must be a path from a vertex in the preimage of $H_1$ to a vertex in the preimage of $H_2$ in $G-e_1-e_2$. This path must contain at least one edge $\edge uv$ that does not participate in any 2-edge-cut, as otherwise it would be a path in the cactus as well. Moreover, $u$ and $v$ must lie in different blobs $B_u$ and $B_v$ of the cactus.

The one that was created last, say $B_u$, must be different from the initial blob. Consider the time when $B_u$ was created in the incremental construction of the cactus. We introduced a cycle to some preexisting blob $B^*$ on which all edges were cut edges, in particular the two cut edges incident to $B_u$. However, the edge $\edge uv$ still connects $B_u$ to the rest of graph, since $B_v$ also exists at this time, a contradiction.
\end{proof}

By applying the techniques of this section, the certifying algorithm for $3$-vertex-connectivity~\cite{Schmidt2013} (which is also based on chain decompositions) can be used to compute the $3$-vertex-connected components of a graph. This has been conjectured in~\cite[p.\ 18]{Schmidt2010b} and yields a linear-time certifying algorithm to construct a SPQR-tree of a graph; we refer to~\cite{Hopcroft1973,Gutwenger2001} for details about $3$-vertex-connected components and SPQR-trees. The full construction can be found in the appendix, section~\ref{3-Vertex Components}.

\section{Conclusion}\label{Conclusion}
We presented a certifying linear time algorithm for 3-edge-connectivity based
on chain decompositions of graphs. It is simple enough for use in a classroom setting and can serve as a gentle introduction to the certifying 3-vertex-connectivity algorithm of~\cite{Schmidt2013}. We also provide an implementation in Python, available at \url{https://github.com/adrianN/edge-connectivity}.

We also show how to extend the algorithm to construct and certify a cactus representation of all 2-edge-cuts in the graph. From this representation the 3-edge-connected components can be readily read off. The same techniques are used to find the 3-vertex-connected components using the algorithm from~\cite{Schmidt2013}, and thus present a certifying construction of SPQR-trees.

Mader's construction sequence is general enough to construct $k$-edge-connected graphs for any $k\geq 3$, and can thus be used in certifying algorithms for larger $k$. So far, though, it is unclear how to compute these more complicated construction sequences. We hope that the chain decomposition framework can be adapted to work in these cases too.


\appendix

\section{Computing a Spanning Subgraph of an Overlap Graph}

We first assume that all endpoints are pairwise distinct. We will later show how to remove this assumption by perturbation.

For every interval $I = [a,b]$ define its set of left and right neighbors:
\begin{align*}
L(I) &= \{ I' = [a',b'] ; a' < a < b' < b \},\\
R(I) &= \{ I' = [a',b'] ; a < a' < b < b' \}.
\end{align*}
If the set of left neighbors is nonempty, let the interval $I' \in L(I)$ with the rightmost right endpoint be the immediate left neighbor of $I$. Similarly, if the set of right neighbors is nonempty, the immediate right neighbor of $I$ is the interval in $R(I)$ with the leftmost left endpoint.

\begin{lemma} The graph $G'$ formed by connecting each interval to its immediate left and right neighbor (if any) forms a spanning subgraph of the overlap graph $G$ and has exactly the same connected components.
\end{lemma}
\begin{proof} Clearly, every edge of $G'$ is also an edge of $G$ and hence
  connected components of $G'$ are subsets of connected components of $G$.

For the other direction, assume $I$ and $I'$ are overlapping intervals that are not connected in $G'$. Then $a < a' < b < b'$, where $I = [a,b] =: I_0$ and $I'=[a',b']$. Let $I_0,I_1,I_2,\ldots$ be such that $I_\ell = [a_\ell,b_\ell]$ is the immediate right neighbor of $I_{\ell - 1}$ for all $\ell \geq 1$. Consider the last $I_\ell$ in this sequence such that $a_\ell < a' < b_\ell < b'$; clearly, such an interval exists, as $I_0$ is such an interval. Then $I'$ is a right neighbor of $I_\ell$, but not the immediate right neighbor of $I_\ell$, as otherwise $I$ and $I'$ would be connected in $G'$. Hence, the immediate right neighbor $I_{\ell +1} =: U =: [c,d]$ of $I_\ell$ exists, is different from $I'$, and must contain $I'$. Thus
\[
	a < c < a' < b < b' < d.
\]
Starting from $I'$ and going to immediate left neighbors, we obtain in the same fashion an interval $U' = [c',d']$ with
\[
	c' < a < a' < b < d' < b'.
\]
We conclude that $U'$ and $U$ overlap, but are not connected in $G'$. 

Consider now a particular choice for the overlapping intervals $I$ and $I'$. We choose them such that the left endpoint of $I$ is as small as possible. However, the left endpoint of $U'$ is to the left of the left endpoint of $I$, and we have derived a contradiction. \qed
\end{proof}

\begin{algorithm}[t]
\caption{Finding a spanning forest of a overlap graph}
\label{alg:spanning_tree}
\begin{algorithmic}
	\Procedure{SP}{$I=\{[a_0,a_0'],\ldots, [a_\ell,a_\ell']\}$}
	\State stack = [\,]
	\State sort $I$ lexicographically in descending order
	\For{$[l,r]$ in $I$}
	  \While{stack not empty and $r>$ top(stack) right endpoint}
	    \State pop(stack)
	  \EndWhile
	  \If{stack not empty and $r\geq$ top(stack) left endpoint}
	    \State connect $[l,r]$, top(stack)
	  \EndIf
	  \State push(stack, $[l,r]$)
	\EndFor
	\State stack = [\,]
	\State sort $I$ lexicographically in ascending order where the key for $[l,r]$ is $[r,l]$
	\For{$[l,r]$ in $I$}
	  \While{stack not empty and $l$<top(stack) left endpoint}
	    \State pop(stack)
	  \EndWhile
	  \If{stack not empty and $l\leq$ top(stack) right endpoint}
	    \State connect $[l,r]$, top(stack)
	  \EndIf
	  \State push(stack, $[l,r]$)
	\EndFor
	\EndProcedure
\end{algorithmic}
\end{algorithm}

It is easy to determine all immediate right neighbors by a linear time sweep over all intervals. We sort the intervals in decreasing order of left endpoint and then sweep over the intervals starting with the interval with rightmost left endpoint. We maintain a stack $S$ of intervals, initially empty. If $I_1 = [a_1,b_1],\ldots,I_k=[a_k,b_k]$ are the intervals on the stack with $I_1$ being on the top of the stack, then $a_1 < a_2 < \ldots < a_k$ and $b_1 <b_2 < \ldots < b_k$, $I_1$ is the last interval processed, and $I_{\ell+1}$ is the immediate right neighbor of $I_\ell$ if $I_\ell$ has right neighbors. If $I_\ell$ does not have right neighbors, $a_{\ell +1} > b_{\ell}$. Let $I = [a,b]$ be the next interval to be processed. Its immediate right neighbor is the topmost interval $I_\ell$ on the stack with $b_\ell > b$ (if any). Hence we pop intervals $I_\ell$ from the stack while $b>b_\ell$ and then connect $I$ to the topmost interval if $b>a_\ell$, and push $I$. The determination of immediate left neighbors is symmetric.

It remains to deal with intervals with equal endpoints. We do so by perturbation. It is easy to see that the following rules preserve the reachability by overlaps and eliminate equal endpoints. E.g., in (4), the two intervals are forced to overlap, so reaching one of the two intervals gives a path to the other; the same reasoning motivates (2) and (3).
\begin{compactenum}[(1)]
\item if a left and a right endpoint are at the same coordinate, then the left endpoint is smaller than the right endpoint.
\item if two left endpoints are equal, the one belonging to the shorter interval is smaller.
\item if two right endpoints are equal, the one belonging to the shorter interval is larger.
\item if two intervals are equal, one is slightly shifted to the right.
\end{compactenum}
In other words, the endpoints of an interval $I_i=[a,b]$ are replaced by $((a,-1,b-a,i)$ and $(b,1,b-a,i))$ and comparisons are lexicographic. The perturbation need not be made explicitly, it can be incorporated into the sorting order and the conditions under which edges are added, as described in Algorithm~\ref{alg:spanning_tree}.

\section{Computing all 3-Vertex-Connected Components}\label{3-Vertex Components}

A pair of vertices $\{x,y\}$ is a separation pair of $G$ if $G - x - y$ is disconnected. Similar to the edge-connectivity case, it suffices to compute all vertices that are contained in separation pairs of $G$ in order to compute all $3$-vertex-connected components of $G$. We assume that $G$ is $2$-vertex-connected and has minimum degree $3$.

For a rooted tree $T$ of $G$ and a vertex $x \in G$, let $T(x)$ be the subtree of $T$ rooted at $x$. The following lemmas show that separation pairs can only occur in chains. Weaker variants of Lemma~\ref{lem:separationpair} can be found in~\cite{Hopcroft1973,Vo1983,Vo1983a}.


\begin{lemma}\label{lem:separationpair}
Let $T$ be a DFS-tree of a $2$-connected graph $G$ and $\calC$ be a chain decomposition of $G$. For every separation pair $\{x,y\}$ of $G$, $x$ and $y$ are contained in a common chain $C \in \calC$.
\end{lemma}
\begin{proof} The following simple observation will be useful. Let $r$ be the root of $T$ and let $x \neq r$ be any vertex. Then for every $t \in T(x) - x$, there is a path $P$ from $t$ to a vertex $s \in G - T(x)$ such that $P$ consists only of vertices in $T(x)-x \cup s$.

We first prove that $x$ and $y$ are \emph{comparable} in $T$, i.e., contained in a leaf-to-root path of $T$. Assume they are not. Then $G-x-y$ consists of at most three connected components: one connected component containing the least common ancestor of $x$ and $y$ in $T$, and the at most two connected components that contain the proper descendants of $x$ and $y$, respectively. According to the observation above,
these components coincide, contradicting that $\{x,y\}$ is a separation pair.

Let $x'$ be the child of $x$ in $T$ that lies on the path $y \rightarrow_T x$. Clearly, if $x' = y$, the chain containing the edge $xy$ is a common chain containing $x$ and $y$. Otherwise, $x' \neq y$. If $x = r$, then there is a back-edge $rt$ such that $t \in T(y)$, according to the fact that $G - r$ is connected by $T - r$ and due to the observation above. 
This back-edge $rt$ implies that the first chain $C$ that traverses a vertex of $T(y)$ starts at $r$ and, hence, contains $x$ and $y$.

In the remaining case, $x' \neq y$ and $x \neq r$. Let $st$ be a back-edge that connects an ancestor $s$ of $x$ with a descendant $t$ of $x'$ (possibly $x'$ itself) such that $s$ is minimal; this edge $st$ exists, since $G$ is $2$-vertex-connected. According to~\cite{Schmidt2013a}, $C_1$ is the only cycle in $\calC$ and it follows that $s < x$. If $t \in T(y)$, the first chain $C$ in $\calC$ that contains such a back-edge contains $x$ and $y$ and, hence, satisfies the claim. Otherwise, $t$ is a vertex in $T(x')-T(y)$. Due to the back-edge $st$, $G-x-T(y)$ is contained in one connected component of $G-x-y$. According to the observation above 
(applied on $y$), $\{x,y\}$ can form a separation pair only if $y$ has a child $y'$ such that all back-edges that end in $T(y')$ start either in $T(y)$ or at $x$. Since $G$ is $2$-connected, there must be a back-edge from $x$ to $T(y')$. The first chain $C$ in $\calC$ containing such a back-edge gives the claim, as it contains $x$ and $y$.
\end{proof}

Similar to edge-connectivity, the connected components of the overlap graph for $C_i$ represent all vertices in separation pairs that are contained in $C_i$. The connected components of the overlap graph can be computed efficiently~\cite[Lemma~51]{Schmidt2013}. After finding all these vertices for $C_i$, a simple modification allows the algorithm in~\cite[p.\ 508]{Schmidt2013} to continue, ignoring all previously found separation pairs: For every separation pair $\{x,y\}$, $x < y$, that has been found when processing $C_i$, there is a vertex $v$ strictly between $x$ and $y$ in $C_i$. Furthermore, by doing a preprocessing~\cite[Property~B, p.\ 508]{Schmidt2013} one can assume that $t(C_i) \rightarrow_T s(C_i)$ also has an inner vertex $w$. We eliminate every separation pair $\{x,y\}$ after processing $C_i$ by simply adding the new back-edge $vw$ to $G$. As the new chain containing $vw$ is just an edge, this does not harm future processing steps.

According to Lemma~\ref{lem:separationpair}, this gives all vertices in the graph that are contained in separation pairs. The $3$-vertex-connected components can then be computed in linear time by iteratively splitting separation pairs and gluing together certain remaining structures, as shown in~\cite{Hopcroft1973,Gutwenger2001}.


\end{document}